\documentclass[journal]{article}

\usepackage{arxiv}

\usepackage{hyperref}
\usepackage{graphicx}          
\usepackage{amsmath} 
\usepackage{amssymb}  

\usepackage{amsthm}
\usepackage{color,soul}
\usepackage{flushend}
\usepackage{cite}
\usepackage{enumitem}
\usepackage{mathtools}
\usepackage{xcolor}
\usepackage{array,multirow}
\usepackage[ruled,linesnumbered]{algorithm2e}

\allowdisplaybreaks
\newtheorem{thmm}{Theorem}

\newtheorem{remm}{Remark}
\newtheorem{assm}{Assumption}

\usepackage{authblk}
\title{\LARGE \bf
A Cyberattack Detection-Isolation Scheme For CAV \\ Under Changing Driving Environment
} 

\author[1]{Sanchita Ghosh}
\author[2]{Nutan Saha}
\author[1]{Tanushree Roy}
\affil[1]{Department of  Mechanical Engineering, Texas Tech University, Lubbock, TX 79409, US. Emails:~{\tt\small sancghos@ttu.edu, tanushree.roy@ttu.edu}.}
\affil[2]{ Department of Electrical Engineering, Veer Surenda Sai University of Technology, Odisha 768018, India. Emails:~{\tt\small nsaha\_ee@vssut.ac.in}. }
\begin{document}

\maketitle
\thispagestyle{empty}
\pagestyle{empty}


\begin{abstract}
Under a changing driving environment, a Connected Autonomous Vehicle (CAV) platoon relies strongly on the acquisition of accurate traffic information from neighboring vehicles as well as reliable commands from a centralized supervisory controller through the communication network. Even though such modalities are imperative to ensure the safe and efficient driving performance of CAVs, they led to multiple security challenges. Thus, a cyberattack on this network can corrupt vehicle-to-vehicle (V2V)  and vehicle-to-infrastructure (V2I) communication, which can lead to unsafe or undesired driving scenarios. Hence, in this paper, we propose a cyberattack detection-isolation algorithm comprised of a unified V2V and V2I cyberattack detection scheme along with a V2I isolation scheme for CAVs under changing driving conditions. The proposed algorithm is constructed using a bank of residual generators with Lyapunov function-based performance guarantees, such as disturbance-to-state stability, robustness, and sensitivity. Finally, we showcase the efficacy of our proposed algorithm through extensive Monte-Carlo simulations using real-world highway and urban driving data. The results show that the proposed algorithm can enhance the cybersecurity of CAVs by detecting cyberattacks on CAV platoons and isolating infrastructure-level traffic manipulation.

\end{abstract}

\section{Introduction}
CAVs  are gradually becoming an essential part of modern transportation by providing improved efficiency, safety, and energy sustainability \cite{litman2017autonomous,dey2015review}. On the other hand, Vehicular Ad-hoc NETwork (VANET) is one of the most explored technologies that can support these V2X communications (vehicle-to-everything i.e. both V2V and V2I) along with several facilities, and protocols to enable reliable operation of these Intelligent Transportation Systems (ITS) under changing traffic environments \cite{usdata,wyk2020} and can transmit local information including emergency warnings \cite{vanet}. While the incorporation of this safety-critical information is crucial for effective platoon control, this also enhances the potential for cyberattacks on these networks \cite{buinevich2019forecasting, survey}. Thus, the task of detecting these cyberthreats in CAVs is crucial and such detection is especially challenging when the mode of operations is changing \cite{sanchita_CCTA2023security}.

Most importantly, VANETS are utilized by vehicle platoons with cooperative adaptive cruise control (CACC) to ensure string stability. Here string stability refers to the reduction of traffic density disturbances during its upstream propagation \cite{swaroop1998intelligent}. To achieve this, each vehicle in the platoon relies on the V2X communication and on its onboard controller. The latter receives data from the preceding vehicle (V2V communication) and from the infrastructure (V2I communication) through a wireless communication network. Using this data along with the onboard sensor measurements (e.g. LiDAR, RADAR), it changes the vehicle's velocity to maintain a desired inter-vehicle distance, namely headway. The V2V communication in CACC enables the vehicle platoon to maintain a smaller headway among each other and thereby increasing the traffic flow \cite{xu2003simulation}. Now, these optimum headways and vehicle velocities are not fixed for different driving or traffic environments. For example, generally on the interstates the platoon headway can be kept smaller with higher velocities since traffic is mostly uninterrupted. Such strategies are used in truck platooning to increase fuel efficiency \cite{balador2022survey}. Conversely, for urban traffic, the headway is mostly larger with slower traffic speeds to avoid collisions during frequent stop-and-go conditions. Thus, such vehicle platoons must operate under different driving ``modes" with changing driving environments \cite{song2018multi}.

Now, the goal for this paper is to ensure safety and string stability while accommodating these changes in the driving environment of the platoon. In particular, the V2I communication facilitates adapting to changing traffic conditions with updated information from global traffic data. Various methods such as machine learning-based classifiers, and dynamic programming are used to determine these driving modes via driving pattern recognition \cite{lin2004driving,song2018multi,driving}.  To ensure the string stability of the platoon and optimal energy consumption, different controller gains are designed for the vehicles under these different driving conditions \cite{jeon2002multi}. The traffic infrastructure uses the V2I network to communicate with the onboard vehicular controllers to switch between these different gains or driving modes. In practice, the infrastructure uses  Global Navigation Satellite System (GNSS) or Visual Positioning Systems (VPS) based data to obtain the trajectory of each vehicle and also  to ascertain the changing driving conditions \cite{song2018multi,driving}. We note here that various other modalities such as traffic reporting services, and social data \cite{tweet,chen2010review} can be also used to ascertain such changes in driving conditions. These redundancies are generally practiced to improve traffic safety \cite{yan2007providing}.

Next, to understand the vulnerabilities of VANETS to cyberattacks, we note the structure of a typical VANET  consists of a Road-Side Unit (RSU) and an On-Board Unit (OBU).  These RSUs are installed on the roadside infrastructure and are responsible for providing local information and actuation signals to vehicles. These can be physically accessed by adversary employees during maintenance to get physical data and tamper the hardware \cite{hasrouny2017vanet}.
Meanwhile, OBUs are installed on vehicles  to receive the sensor data and the transmitted information from RSUs through wireless devices.  Moreover, VANET can be configured to change communication topology between communicating vehicles as well as with the local RSUs \cite{hasrouny2017vanet,sdn}. This further enhances the susceptibility of platoons to unsafe or undesirable operations.  Occasionally, under VANET, a Dedicated Short Range Communication (DSRC) protocol is used and the U.S. Federal Communications Commission has allocated 75 MHz of licensed spectrum in the 5.9 GHz band for ITS \cite{dsrc}. Even though this short-range communication improves VANET security against long-range and stationary attackers, yet it fails to secure the vehicular network completely \cite{bai2010toward}.

Hence, in this paper, we consider the possibility of cyberattacks through  the compromised V2V and V2I wireless networks.  While  attacks on both of these networks can lead to unsafe or undesirable scenarios such as rear-end collisions \cite{wang2020modeling}, CACC disengagement \cite{gunter2020commercially}, traffic congestion \cite{engoulou2014vanet}, etc., the V2I communication safety and privacy are remarkably more significant since every vehicle discloses its identity and tracking data \cite{al2016security}. Consequently, while an  attack on the V2V channels explicitly impacts individual or adjacent vehicles in the platoon, an attack on the V2I network compromises  the communications between all vehicles and critical infrastructure components (such as ramp metering, RSUs, etc.) and hence poses a potential threat for a wide-spread disruption in the whole traffic system \cite{roy2021cyber}. For example, the adversary can use the V2I network to send tailor-made signals to multiple CAV platoons to create scenarios like bottleneck traffic cognition, stop-go traffic, and so on \cite{reilly2016creating}. Moreover, (under a changing driving environment scenario) a cyberattack on the V2I channels mimics a compromised switching attack scenario and such an attack exhibits a high likelihood of leading the system to an unsafe operating region \cite{ghosh2023security}.  Thus, it is important to isolate the attacks on the V2I channel, if possible \cite{islam2018cybersecurity}. To address these threats, we categorize cyberattacks as  V2V and V2I attacks \cite{v2vattack,boddupalli2022resilient,mashrur,islam2018cybersecurity}. Additionally, the V2X attack denotes either concurrent  or asynchronous V2V-V2I attacks \cite{han2023secure}.  Thus we propose a unified diagnosis scheme to first detect the presence of V2V and/or V2I attacks and then, an isolation scheme to distinguish V2I attack scenarios. 

\section{Related Literature And Contribution}

Existing CAV security and cyberattack detection can be broadly classified into (i) information-based, (ii) machine learning-based, and (iii) model-based approaches. In the information-based safety analysis, authors focused mainly on secured networks and data transmission \cite{kalinin2017network,bai2010toward,dataencrypt} along with authentication and plausibility confirmation of data received through vehicular networks \cite{xiao2006detection,trustdata,mundhenk2017security}. 
Similarly, in \cite{dataencrypt}, authors proposed a data-encryption scheme to achieve reliable vehicle platoon operation. To enhance security while ensuring flexibility and programmability of the network, several works focused on software-defined network (SDN) \cite{vanet,sdn, yan2015software}. Separating the centralized control panel from the data panel and the dynamic updating  makes SDN more secure against cyberthreats \cite{yan2015software}. However, when these  informatics-based security measures are breached, the adversary poses severe risks to the vehicle platoon by gaining access to physical control \cite{mundhenk2017security}.

Furthermore, several works focus on machine learning-based techniques to detect cyberattacks on CAVs. For example, a convoluted neural network (CNN) based detector is proposed in \cite{javed2020cnn} to detect multi-source anomalies as well as single-source anomalies in sensor behavior. Thus the proposed solution is not restrictive to the detection in  specific scenarios only. Likewise, the authors proposed a reinforced learning-based data trust model to improve data authenticity in \cite{rmltrust}.  In \cite{mashrur}, authors utilized expectation maximization (EM) and cumulative summation (CUSUM) algorithms to ensure real-time detection of V2I cyberattacks. Additionally, in \cite{boddupalli2022resilient}, different driving environments are considered as a parameter during ML-based model training for the V2V attack detection strategy. However, these approaches rely heavily on the availability of data \cite{wyk2020} and it is particularly difficult to train such classifier-based detectors when data generated by CAVs is limited and incomplete \cite{bangui2022hybrid}. With increasing connectivity among smart infrastructures, the potential  attack vectors increase proportionally, which in turn increases the occurrence of novel forms of attack on CAVs. Thus, limited data severely handicaps machine learning-based strategies \cite{bangui2022hybrid}, particularly in the context of new or unknown attacks \cite{close}.

In contrast, model-based approaches employ state estimation of the vehicle platoon using mathematical models of CAVs, to detect any deviation in the behavior of the CAVs from the nominal operation and consequently detect cyberattacks \cite{stateestimation}. For instance,  \cite{kremer2020kf} proposed a Kalman filer-based (KF) detection scheme for an ego vehicle under a compromised controller by utilizing the measurements from  neighboring vehicles. Likewise, an extended KF observer is used to detect cyberattacks under communication delay in \cite{kf}. Additionally, adaptive resilient observer-based distributed detectors are used to isolate between compromised and healthy vehicle sensors \cite{he2021distributed}. A per-vehicle detection scheme consisting of two ellipsoidal sets, namely a state prediction set and a state estimation set, is proposed in \cite{estimation}. This method ensures the detection of both V2V and V2I attacks. Moreover, along with the system model and control-based tools, authors utilized  the redundancies among the physical signals and the social signals received from users' mobile and social media to detect cyberattacks on connected transportation systems in \cite{roy2021socio}. Similarly, sensor information redundancies were exploited to isolate the compromised vehicles under Syb attacks in \cite{yang2021secure}. \cite{biron2018real} proposed a detection algorithm scheme for Denial-of-Service (DoS) attacks on the communication network by means of sliding mode and adaptive observers for individual vehicles and estimated the impact of the attack on CAV. Finally, an integrated framework combining CNN and KF was proposed in \cite{wyk2020} to detect anomalous sensors for CAV platoons.


Nevertheless, these works considered a single control law for all driving conditions, although multi-mode control strategy in CAVs is essential to ensure optimal operation such as energy efficiency, vehicle safety, etc. under changing driving environments \cite{song2018multi,zhai2018switched}.   Moreover, these works did not consider isolating V2I attack scenarios,
where the latter introduces a variety of vulnerabilities, including exposed infrastructure data access and large-scale resource disruption \cite{ roy2021cyber}.  
Furthermore, identification of single or multiple compromised vehicles is often unaddressed in many of these works, which can enable the authorities to respond with proper countermeasures knowing the source and nature of the attacks.
 Therefore, to address these research gaps, in this paper our contributions are as follows: 
\begin{itemize}
    \item We propose a detection-isolation algorithm based on residual generation to detect both V2V and V2I (i.\,e. V2X) attacks in a CAV platoon traveling under different modes depending on the changing driving environment. 
    \item The proposed algorithm can also isolate the attacks on the V2I channel without restrictions on the attack policies.
    \item The proposed algorithm can identify the compromised vehicles and thus track down the source of the attack.
    \item The efficacy of the algorithm is tested using realistic driving patterns from \href{https://ops.fhwa.dot.gov/trafficanalysistools/ngsim.htm}{NGSIM} (Next-Generation Simulation)  and \href{https://www.epa.gov/emission-standards-reference-guide/epa-urban-dynamometer-driving-schedule-udds}{UDDS} (Urban Dynamometer Drive Schedules) data.   
\end{itemize}


The rest of the paper is organized as follows. In Section \ref{2}, we present the problem formulation including the closed-loop system model of the CAV under consideration. Section \ref{3} presents the proposed V2X detection scheme and mathematical analysis to obtain necessary and sufficient conditions for performance guarantees. Next, a V2I isolation scheme is presented in Section  \ref{ie}. The simulation results showing the effectiveness of our proposed algorithm to detect V2X attacks as well as to isolate V2I attacks are presented in Section \ref{4}. Finally, we present the concluding remarks  in Section  \ref{5}.

\textbf{Notation:}  In this paper, we used $\mathbb{R}$ to denote the set of natural numbers. The Euclidean norm of a vector ${\lVert x\rVert}_2=\sqrt{x^Tx}$, andthe $\mathcal{L}_2$ norm of a vector-valued function ${\lVert x\rVert}_{\mathcal{L}_2}:=\sqrt{\left(\int_0^\infty {\lVert x\rVert}^2_2 dt\right)}$.  The absolute value of the variable $a$ is denoted as  ${\lvert a\rvert}$. We used $I_a$ to denote an identity matrix. To denote the minimum and maximum eigenvalue of a matrix A, we used $\lambda_{\min[A]}$ and $\lambda_{\max[A]}$ respectively.  $\Dot{f}$ denotes derivative with respect to $t$ i.e. $\dot{f} = \frac{df}{dt}$. For two vectors $x$, $y$, and a matrix $A$, we define a binary quadratic function $q(x,A,y) = x^T A y + y^T A^T x$.

\section{Problem Formulation} \label{2}
Let us consider a vehicle platoon under CACC with V2V and V2I communication. The vehicle platoon runs with different operating modes  to cope with changing driving environments. The leader vehicle drives in response to the surrounding driving environment. We note here that the leader vehicle is either human-driven or semi-autonomous which is auspicious in practical implementation \cite{yao2020managing}. A supervisory controller monitors the trajectory, position, and acceleration of the leader vehicle to decide upon the appropriate operational  mode for the rest of the vehicle platoon \cite{turri2016cooperative}.
It communicates this information to the onboard vehicular controllers of the follower vehicles to switch among different ``modes".  We note here that each mode corresponds to specific controller gains. These different gains along with the V2V and V2I data ensure the string stability of the platoon. In this paper, we assume that the infrastructure uses GNSS-based data to obtain the trajectory and position data of each vehicle. \mbox{Fig. \ref{fig:blockdiagram}} shows the block diagram of our problem framework.



\begin{figure}[ht!]
    \centering
    \includegraphics[trim = 0mm 0mm 0mm 0mm, clip,  width=0.9\linewidth]{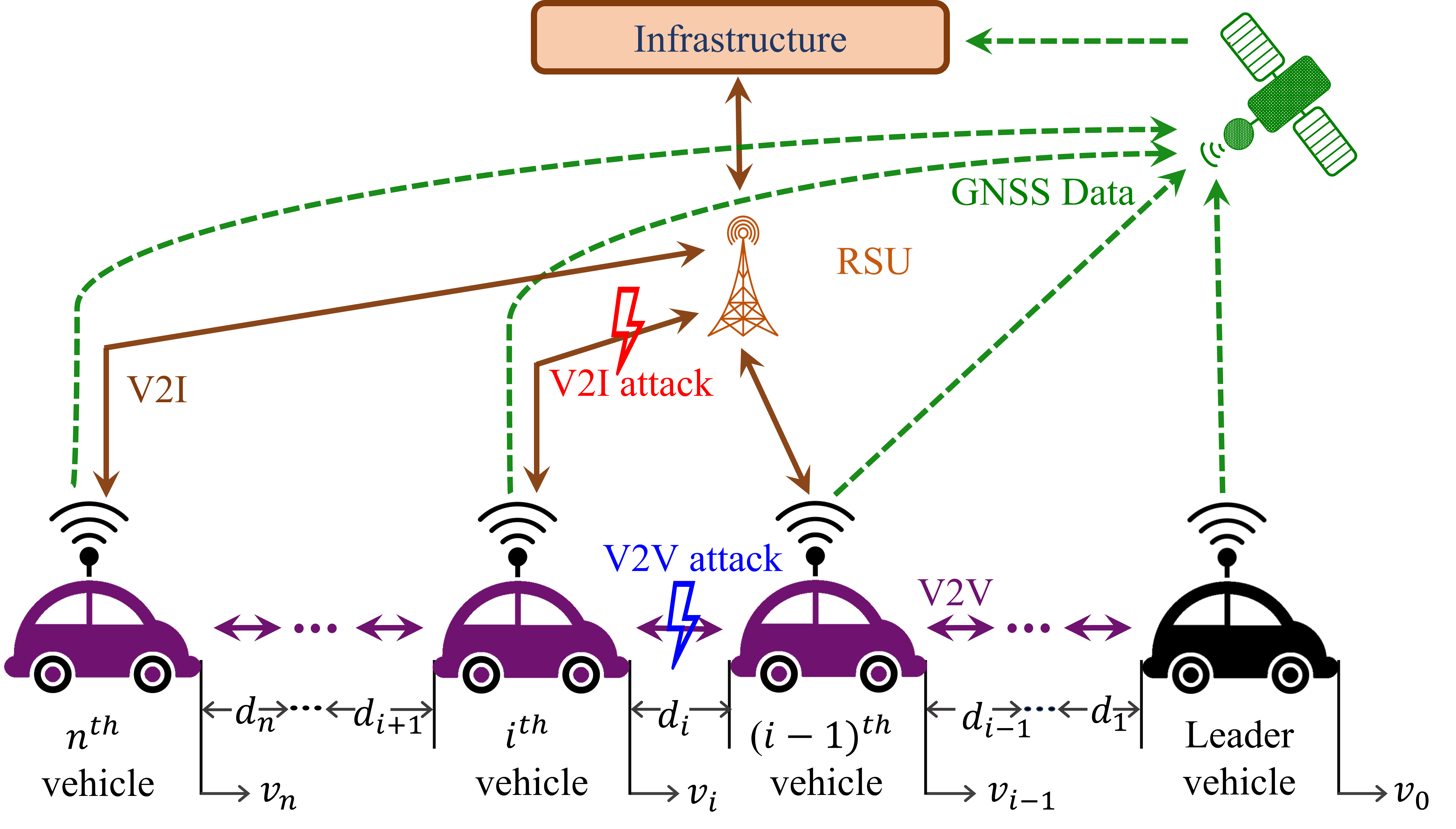}   
    \caption{Block diagram of the vehicle platoon.}
    \label{fig:blockdiagram}
\end{figure}

\begin{assm}
    In this problem setting, we assume that the traffic measurements from the GNSS or VPS are not simultaneously compromised along with the communication network. Additionally, under nominal conditions, we are assuming a delay-free communication network. Generally, consideration of delays in traffic measurements results in a relatively slower detection time \cite{roy2020secure} and will be considered in our future work.
\end{assm}
\begin{remm}
     The human-driven or semi-autonomous leader vehicle of the CAV platoon is not susceptible to cyberattacks. Yet, if the leader vehicle is a CAV, we will still be able to extend our framework and detect-isolate cyberattacks on the leader vehicle (in the near future), as long as we can collect the positional data of the vehicle preceding the leader. In the unlikely event that such information is unavailable, our algorithm in its present form will be unable to detect the cyberattack on the leader and such scenarios remain beyond the scope of our current work. 
\end{remm}
\noindent

\subsection{Closed-loop vehicle dynamics}
We consider a car following (CF) vehicle platoon with a string of $n + 1$ vehicles. First, let us define the state vector for the leader vehicle as $\overline{z}_0 = \begin{bmatrix}
    x_0 & v_0
\end{bmatrix}^T$ where $x_0, v_0$ are respectively the position and velocity of the leader vehicle. The vehicle dynamics of the leader vehicle can be written as
   $ \dot{\overline{z}}_0 = A_0 \overline{z}_0 + B_0 u_0,\, t\geqslant 0,$
 Here the state matrix $A_0 \in \mathbb{R}^{2 \times 2}$ is defined as $A_0 = \begin{bsmallmatrix}
    0 & 1 \\
    0 & 0
\end{bsmallmatrix}$. The input matrix  $B_0 \in \mathbb{R}^{2 \times 1}$ is defined as $B_0 = \begin{bsmallmatrix}
    0 & 1 
\end{bsmallmatrix}^T$. The input $u_0 \in \mathbb{R}^1$ is defined as $u_0 = a_0$. $a_0$ is the acceleration for the leader vehicle. Let $i\in \{0, 1, \cdots, n\}$ represent the vehicle number in the  platoon where $i = 0$ represents the leader vehicle. The rest $n$ numbers of follower vehicles are expected to follow their preceding vehicle at a desired headway $h_{d_i}$, which is defined as 
$
    h_{d_i} = s_i + \mathcal{T}_\alpha v_i, \quad 1 \leqslant i \leqslant n,
$
where $s_i$ is the standstill distance for $i^{th}$ vehicle and $\mathcal{T}_\alpha$ is the time headway during $\alpha^{th}$ operating mode. $v_i$ is the velocity of the $i^{th}$ vehicle. Here the piece-wise constant switching signal $\alpha : [0,\infty) \rightarrow \mathcal{M} = [1, 2, \cdots, m]$  from infrastructure indicates the  operational mode of the system. The platoon is assumed to be homogeneous and therefore the time headway $\mathcal{T}_\alpha$ is not dependent on the vehicle $i$. The actual headway  between vehicle $i$ and the preceding vehicle $i-1$ is denoted as $h_i$. 
\noindent
Now, let us define state vector $z_i = \begin{bmatrix} h_i & v_i & a_i \end{bmatrix}^T$, where $a_i$ is the acceleration of the $i^{th}$ vehicle. We then re-write the state vector of the leader vehicle as $z_0 = \begin{bmatrix}
    0 & v_0 & a_0
\end{bmatrix}^T$. Then  the closed-loop vehicle dynamics of the $i^{th}$ vehicle becomes 
\begin{align}
    &\dot{z}_i = A z_i + D z_{i-1} + B u_i + \overline{w}_i, \hspace{1mm} y_i = C z_i, \hspace{1mm} 1\leqslant i \leqslant n \label{xi} \\
    &\dot{u}_i = - k_{u_\alpha} u_i + A_{c_\alpha} z_i + D_{c_\alpha} z_{i-1} - H_{\alpha} s_i + \frac{\theta_i}{\mathcal{T}_\alpha} u_{i-1}. \label{ui}
\end{align}
 Here $A, D \in \mathbb{R}^{3\times3}, B \in \mathbb{R}^{3\times 1}, A_{c_\alpha}, D_{c_\alpha} \in \mathbb{R}^{1\times3}$, and  $ H_{\alpha} \in \mathbb{R}^{1}$ are considered identical for all vehicles following the homogeneous platoon assumption and are defined as 
\begin{align} 
    &A = \begin{bmatrix}
        0 & -1 & 0 \\
        0 & 0 & 1 \\
        0 & 0 & -\frac{1}{\Sigma}
    \end{bmatrix},\,
    A_{c_\alpha} =\frac{k_\alpha}{\mathcal{T}_\alpha} \begin{bmatrix}
        1 & -\mathcal{T}_\alpha & 0\\
        0 & -1 & -\mathcal{T}_\alpha \\
        0 & 0 & \frac{\mathcal{T}_\alpha}{\Sigma} - 1
    \end{bmatrix},\nonumber \\
     & D = \begin{bmatrix}
        0 & 1 & 0 \\
        \textbf{0} & \textbf{0} & \textbf{0} 
    \end{bmatrix},\,
    D_{c_\alpha} = \frac{k_\alpha}{\mathcal{T}_\alpha} \begin{bmatrix}
        0 & \textbf{0}  \\
        \textbf{0} & I 
    \end{bmatrix}, \nonumber \\
    &B = \begin{bmatrix}
        0 & 0  &  \frac{1}{\Sigma}
    \end{bmatrix}^T, H_{\alpha} = \frac{k_\alpha}{\mathcal{T}_\alpha}\begin{bmatrix}
        1 & 0 & 0
    \end{bmatrix}^T.     
\end{align} 
\noindent The parameter $\Sigma$ is the time constant to represent the engine dynamics of the vehicle, $y_i \in \mathbb{R}^q (q \leqslant 3)$ is the output and $u_i \in \mathbb{R}^1$ is the controlled input. $C \in \mathbb{R}^{q\times 3}$ is the output matrix. $\overline{w}_i \in \mathbb{R}^3$ is the uncertainty.  $k_\alpha \in \mathbb{R}^{1 \times 3}$ is the controller gain vector defined as $ k_\alpha = \begin{bmatrix}
    k_{p_\alpha} & k_{d_\alpha} & k_{dd_\alpha}
\end{bmatrix}$ and \mbox{$k_{u_\alpha} = \frac{1}{\mathcal{T}_\alpha}  + \frac{k_{dd_\alpha}}{\Sigma} $}. The function $\theta_i : [0,\infty) \rightarrow \{ 0, 1 \}$ represents the  V2V communication where the values $0$ and $1$ indicate ceased or  intact communication respectively.

\noindent
To analyze the closed-loop vehicle dynamics, let us define an augmented state vector $\zeta_i = \begin{bmatrix}
    z_i^T & u_i
\end{bmatrix}^T$ which yields
\begin{equation}\label{single_aug}
    \dot{\zeta}_i = \mathcal{A}_\alpha \zeta_i + \mathcal{D}_{\theta \alpha} \zeta_{i-1} - \mathcal{H}_{\alpha} s_i + w_i, \hspace{3mm} y_i = \mathcal{C} \zeta_i,
\end{equation}
where  $\mathcal{A}_\alpha,\mathcal{D}_{\theta \alpha} \in \mathbb{R}^{4\times4}$,  $\mathcal{C} \in \mathbb{R}^{q \times 4}$, and $\mathcal{H}_{\alpha}, w_i \in \mathbb{R}^{4\times 1}$ are defined as\begin{align}
    \mathcal{A}_\alpha =& \begin{bmatrix}
        A & B \\
        A_{c_\alpha} & -k_{u_\alpha}
    \end{bmatrix}, \mathcal{D}_{\theta \alpha} = \begin{bmatrix}
        D & \textbf{0}\\ D_{c_\alpha} & \frac{\theta_i}{\mathcal{T}_\alpha}
    \end{bmatrix}, \mathcal{C} = \begin{bmatrix}
    C^T \\ \textbf{0}
\end{bmatrix}^T,\nonumber \\ &\mathcal{H}_{\alpha} = \begin{bmatrix}
        \textbf{0} & H_{\alpha}
    \end{bmatrix}^T, w_i = \begin{bmatrix}
    I & \textbf{0}
\end{bmatrix}^T \overline{w}_i. 
\end{align}

\begin{assm}
The switching between the operational modes of the CAV platoon is assumed to be significantly slow since the driving environment changes infrequently. In particular, we assume that such switching time is more than the average dwell time (ADT) \cite{hespanha1999stability} and the platoon/vehicles remain stable under the designed controller gains for each mode. 
\end{assm}

\subsection{Cyberattack policy and compromised platoon dynamics}

\begin{figure}[h!]
    \centering
    \includegraphics[trim = 0mm 0mm 0mm 0mm, clip,  width=0.75\linewidth]{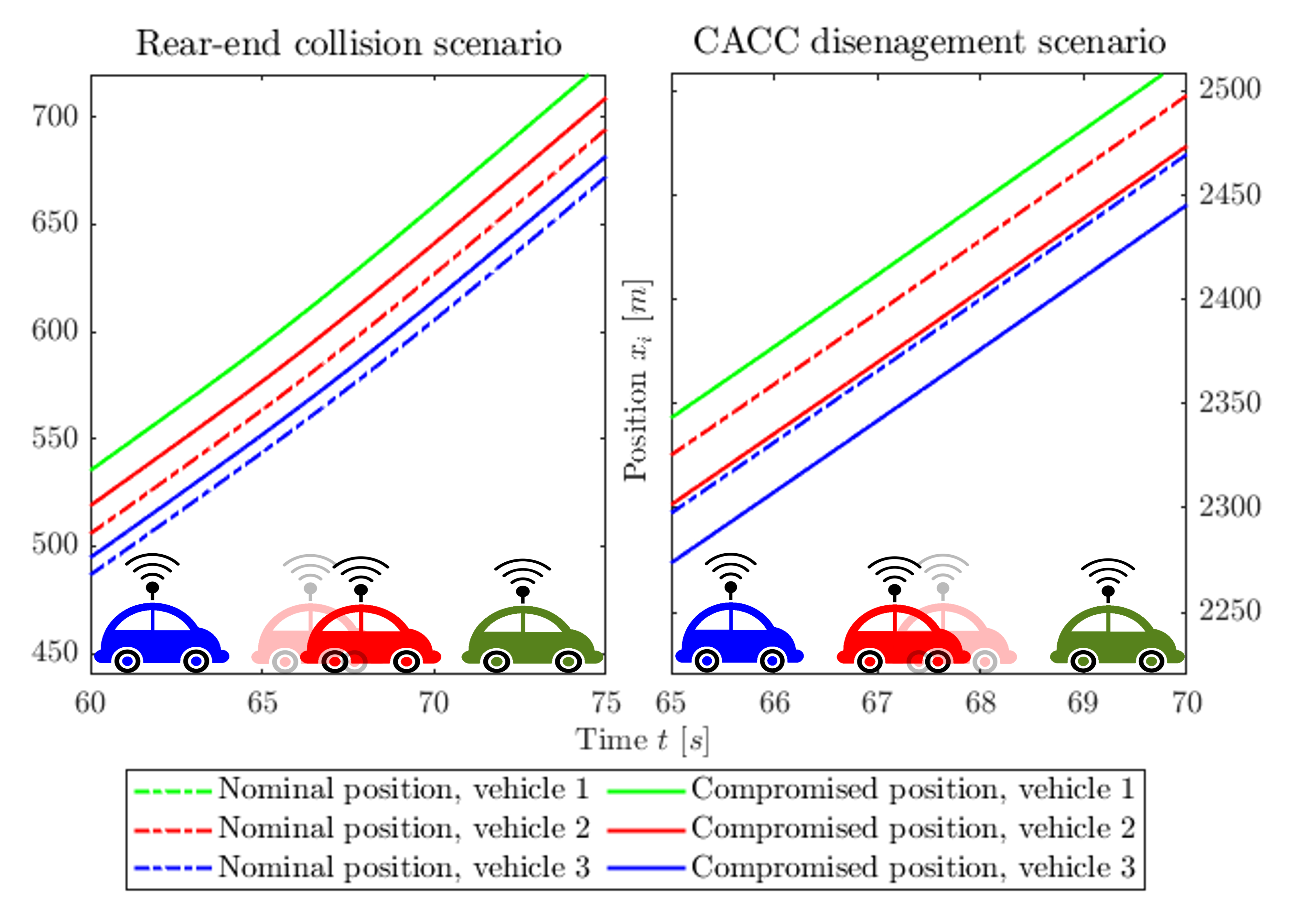}   
    \caption{Objectives of the adversary.}
    \label{fig:combined}
\end{figure}
In this framework, we consider that the adversary can simultaneously attack single or multiple follower vehicles of the platoon either through the V2V or through the V2I communication channel. An attack on both the V2V and the V2I communication (i.\,e. V2X) is possible as well. Here the primary objective of the adversary is to drive the platoon to unsafe or undesired conditions. In unsafe conditions, vehicles may run too close to each other (too small a headway distance) during the stop-go traffic condition. This can lead to  a high probability of collisions as shown in the left plot of \mbox{Fig. \ref{fig:combined}}. On the other hand, a compromised vehicle may trail off if the headway distance is too large and the velocity is too small compared to the preceding vehicle, leading to disengaged CACC \cite{gunter2020commercially}. The right plot of \mbox{Fig. \ref{fig:combined}} exhibits this scenario. 

The adversary can opt for  several active network attack policies such as DoS,  distributed DoS (DDoS), blackhole, false-data-injection (FDI), and replay attacks to manipulate the platoon control policy and/or measurement \eqref{single_aug}. They particularly affect the control input dynamics in \eqref{ui} \cite{duo2022survey}. Such attacks can be formulated with the switching command signal $\alpha_i$ and the V2V communication function  $\theta_i $ along with the preceding vehicle's control input $u_{i-1}$ to mathematically represent these different attack policies to better understand their impacts in the context of our framework. Furthermore,  the compromised scenarios are denoted by $\Tilde{\alpha}_i$ and $\Tilde{\theta}_i \Tilde{u}_{i-1}$. We define two Boolean scalars ${\delta_\alpha}$ and ${\delta_\theta}$ such that  ${\delta_\alpha},{\delta_\theta} = 1$ indicate the activation of  V2V and V2I attacks,  respectively.

\textbf{DoS / DDoS Attack:} During DoS attacks,  the compromised vehicle  stops receiving any updates from the infrastructure or\,/\, and the preceding vehicle and hence,  continues to run with the previous information \cite{hasrouny2017vanet}.   Mathematically we define this scenario for the $i^{th}$ vehicle as 
\begin{align}
   &\text{V2V attack: } \Tilde{\theta}_{i} \Tilde{u}_{i-1}= \theta_i \left( u_{i-1} - {\delta_\theta} \left[ u_{i-1} - u_{i-1} (t_0) \right] \right), \\
   &\text{V2I attack: } \Tilde{\alpha}_i = \alpha_i - {\delta_\alpha} \left[ \alpha_i - \alpha_i (t_0) \right], 
\end{align} where  the attacks begin at time $t_0$.
Multiple vehicles are corrupted during the DDoS attack, i.\,e.  these V2X communications can be halted  for more than one vehicle.

\textbf{False Data Injection (FDI):} In this scenario, the adversary provides the compromised vehicle/s with the information to switch to a wrong controller mode during a V2I attack. Conversely, during a V2V attack the adversary provides wrong $u_{i-1}$ information. Hence, we define  this scenario as follows.
\begin{align}
    & \text{V2V attack: }\Tilde{\theta}_{i} \Tilde{u}_{i-1} = \theta_{i} \left( u_{i-1} - {\delta_\theta} \left[ u_{i-1} - \overline{u}_{i-1} \right] \right), \\
    &\text{V2I attack: }\Tilde{\alpha}_i = \alpha_i - {\delta_\alpha} \left[ \alpha_i - \overline{\alpha}_i \right]; \,\overline{\alpha}_i \in \mathcal{M} \smallsetminus \{\alpha_i\}.
\end{align} 
An example of FDI  is the  Blackhole attack where the compromised vehicle ceases any upstream information flow \cite{chowdhury2020security}.

\textbf{Replay Attack:} This is a two-phase attack. First, the adversary stores the switching command signals $\alpha$ or\,/\,and the preceding vehicle's control input $u_{i-1}$ for a time interval $\left[t_0-\Delta \tau, t_0\right]$. 
Then in the next phase, the adversary starts to replay the stored commands 
in a loop. Thus, for $t_a = t - \Delta \tau$, the final attack can be represented as
\begin{align}
    & \text{V2V attack: }\Tilde{\theta}_{i} \Tilde{u}_{i-1}\!= \theta_{i} \left( u_{i-1} \!- {\delta_\theta} \left[ u_{i-1} \!- u_{i-1} ( {t}_a) \right] \right)\!,  \\
    &\text{V2I attack: }\Tilde{\alpha}_i = \alpha_i - {\delta_\alpha} \left[ \alpha_i - \alpha_i ( t_a) \right].
\end{align}

\noindent
Hence, we define the V2V and V2I attacks concisely as
\begin{align}
    &\Tilde{\theta}_i \Tilde{u}_{i-1}= \theta_i u_{i-1} - \Delta_\theta; \quad \Tilde{\theta}_i : [0,\infty) \rightarrow \{0,1\}, \\
    &\widetilde{\alpha}_i = \alpha_i - \Delta_\alpha; \quad \widetilde{\alpha}_i : [0,\infty)\rightarrow \mathcal{M} = \{1,2,\cdots,m\}. \label{v2i_tilde}
\end{align} Here, $\Delta_\alpha, \Delta_\theta$ are defined such that they indicate the presence of any cyberattacks when $\delta_\alpha, \delta_\theta \neq 0$.
Now, comparing \eqref{v2i_tilde} with standard compromised switching attack definition \cite{ghosh2023security}, we can deduce that the V2I cyberattacks are in fact compromised switching attacks for the platoon, and thus, the compromised vehicles run with mismatched controllers.  On the other hand, during a V2V attack the compromised vehicles run with the correct controller, however track a wrong desired acceleration. Accordingly, the V2I attacks appear as multiplicative anomalies in the closed-loop vehicle dynamics in terms of $\mathcal{A}_{\widetilde{\alpha}},\mathcal{D}_{{\theta} \widetilde{\alpha}}, $  $\mathcal{H}_{\widetilde{\alpha}}$ and the V2V attacks appear as additive anomalies as $\mathcal{D}_{\Tilde{\theta} {\alpha}} \Tilde{u}_{i-1} $.

Thus, we utilize the standard methods   \cite{dingbook} to convert the multiplicative anomalies to additive anomalies. However, the three possible cases namely only V2V, only V2I, and simultaneous V2V-V2I attacks affect the platoon dynamics differently. Therefore,  the compromised platoon dynamics is  re-written as 

\begin{align} \label{zeta}
     \dot{\zeta}_i& = \mathcal{A}_\alpha \zeta_i + \mathcal{D}_{\theta \alpha} \zeta_{i-1} - \mathcal{H}_{\alpha} s_i + E_{\theta i} f_{\theta i} + E_{\alpha i} f_{\alpha i}  \nonumber \\ & + E_{\theta \alpha i} f_{\theta \alpha i} + w_i, \quad \zeta_i(0)=\zeta_{i0}, \quad t\geqslant 0.
\end{align}
 $ E_{\theta i}, E_{\alpha i}, E_{\theta \alpha i}$ are the known matrices and $f_{\theta i}, f_{\alpha i}, f_{\theta \alpha i}$ are the unknown time functions representing respectively only V2V, only V2I,  simultaneous V2V-V2I attacks such that
$\mathcal{A}_{\widetilde{\alpha}} \zeta_i + \mathcal{D}_{\Tilde{\theta} \widetilde{\alpha}}\Tilde{\zeta}_{i-1} - \mathcal{H}_{\widetilde{\alpha}} s_i  = \mathcal{A}_\alpha \zeta_i + \mathcal{D}_{\theta \alpha} \zeta_{i-1} - \mathcal{H}_{\alpha} s_i + E_{\theta i} f_{\theta i}  + E_{\alpha i} f_{\alpha i} + E_{\theta \alpha i} f_{\theta \alpha i} .$ In other words $\Delta_\theta \neq 0$  implies $f_{\theta i} \neq 0$, $\Delta_\alpha \neq 0$ implies $f_{\alpha i} \neq 0$, and  $\Delta_\theta, \Delta_\alpha \neq 0$ implies $f_{\theta i}, f_{\alpha i}, f_{\theta \alpha i} \neq 0$. Considering  the closed-loop  vehicle platoon dynamics \eqref{zeta}, we now introduce our  proposed algorithm. 

 \subsection{Cyberattack detection-isolation algorithm}
 \begin{figure}[h!]
    \centering
    \includegraphics[trim = 0mm 0mm 0mm 0mm, clip,  width=0.9\linewidth]{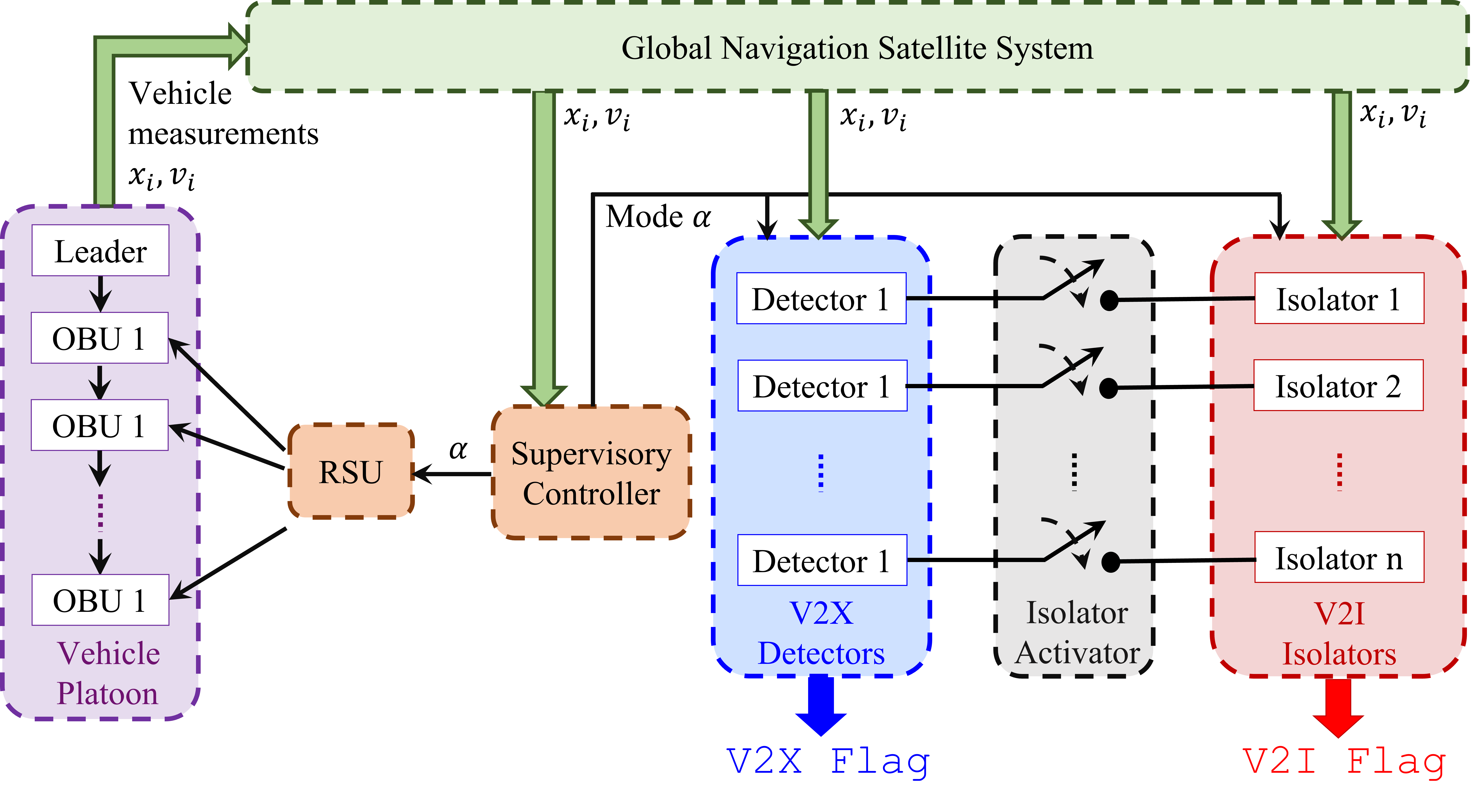}   
    \caption{Schematic block diagram of the detection-isolation algorithm.}
    \label{fig:long_dia}
\end{figure}
The objective of our proposed two-phase algorithm is to detect the presence of any V2X cyberattacks and thereafter isolate the presence of V2I attacks. The first phase is the V2X detection scheme (DS) which monitors individual vehicle performance for any deviation from the desired behaviors. Such deviations are captured by the V2X DS residual and a V2X attack is detected when the V2X DS  residual crosses a threshold. Upon detection of a V2X cyberattack in a vehicle, a V2X flag is generated and the corresponding second phase of the algorithm is activated for that vehicle. This is the V2I isolation scheme (IS). Similar to the first phase, the V2I IS investigates individual vehicle performance to generate a V2I flag if the performance deterioration is due to a V2I attack
and hence the V2I IS residual crosses a threshold.
To achieve this goal, we consider a bank of $n$ number of switched-mode detectors and $n$ number of switched-mode isolators corresponding to the $n$ follower vehicles of the platoon. All of these $2n$ detectors and isolators have $m$ operational modes corresponding to the $m$ operational modes of the vehicle platoon. Similarly, all the $n$ detectors and isolators run with the same mode $\alpha$ at a time and receive the mode information from the supervisory controller. 
The proposed detection-isolation algorithm is shown as a schematic diagram in \mbox{Fig. \ref{fig:long_dia}}.

Furthermore, the V2X DS and the V2I IS should be robust to uncertainties that may arise from the unmodeled system dynamics or measurement noises. Contrarily, they must remain sensitive to cyberattacks that are specifically crafted by an adversary to drive the system towards unsafe or undesired states. Therefore, we designed all the detector and isolator gains to ensure the following performance guarantees \cite{ghosh2023security}: 

    \noindent
    \textbf{Disturbance-to-state stable (DSS):} DSS refers to the condition, where for any driving mode $\alpha$ for a vehicle, the residual for the detector/isolator must remain bounded under bounded V2V and/or V2I attack and under bounded uncertainties.

   \noindent
    \textbf{Robustness:} To ensure robustness, the gains must be chosen such that the effect of the uncertainties on the residual is suppressed by a robustness factor (RF), which must be minimized.
    
\noindent
    \textbf{Sensitivity:} Conversely, we want to choose the gains such that it amplifies the effect of every attack by a  sensitivity factor (SF), which must be maximized.

\begin{remm}
       RF and SF cannot be arbitrarily minimized or maximized simultaneously since their impact on the gain is opposite. For a given gain, if SF is maximized, the algorithm will be highly sensitive to both attacks and uncertainties. This increases the probability of attack detection as well as the probability of invalid flag generation due to uncertainties  leading to false alarms. Conversely, if  the RF is kept small, then the algorithm will be less sensitive and more robust to uncertainties while suffering from a high chance of missed detection in presence of cyberattacks. Consequently, we use an optimization scheme to address this robustness-sensitivity trade-off. Specifically, we utilize the optimization weighting constant to prioritize between these two factors and thus, to obtain an appropriate balance between the robustness and the sensitivity of the algorithm. Lastly, we note that the robustness-sensitivity balance varies with the design criteria of preferred uncertainty tolerance and detectable attack magnitude. It remains unaffected by the complexity and patterns of cyberattacks and thus, will be effective for new and emerging attacks of appropriate magnitude.
       
\end{remm}

\noindent


\section{V2X Detection Scheme} \label{3}
In this section, we first present the V2X detector dynamics and then obtain the design conditions in the form of  linear matrix inequalities (LMIs) to ensure the desired performance.
\subsection{V2X detector dynamics}
Let us first define the augmented detector dynamics \eqref{chihat}, \mbox{$\forall i \in \{ 1, 2, \cdots,n\}$} considering the system model \eqref{zeta}  such that  the detector  $i$  corresponds to the $i^{th}$ vehicle of the platoon.
\begin{equation}\label{chihat}
    \dot{\hat{\zeta}}_i =  \mathcal{A}_\alpha \hat{\zeta}_i + \mathcal{D}_{\theta \alpha} \hat{\zeta}_{i-1} - \mathcal{H}_{\alpha} s_i +\mathfrak{L}_\alpha (y_i-\hat{y}_i),\hspace{3mm}\hat{y}_i = \mathcal{C} \hat{\zeta}_i.
\end{equation} Here $\hat{\zeta}_i \in \mathbb{R}^{4}$ and $\hat{y}_i\in \mathbb{R}^{q}$ are detector state and output respectively.   $\mathfrak{L}_\alpha \in \mathbb{R}^{4 \times q}$ is the  detector gain matrix to be designed. We note here that  $\mathfrak{L}_\alpha$ is the same for each detector due to the homogeneity assumption on the platoon. 

\subsection{Error dynamics and V2X DS residual generation}
The  error between the platoon dynamics and the detector dynamics is defined as $e_i = {\zeta}_i - \hat{\zeta}_i$. Thus, form \eqref{zeta} and \eqref{chihat} the error dynamics can be written as
\begin{align}\label{err}
    \dot{e}_i & = \left[\mathcal{A}_{\alpha} - \mathfrak{L}_\alpha \mathcal{C} \right] e_i + \mathcal{D}_{\theta \alpha} {e}_{i-1}  +E_{\theta i} f_{\theta i}\nonumber \\ & + E_{\alpha i} f_{\alpha i} + E_{\theta \alpha i} f_{\theta \alpha i} +w_i.
\end{align}

\noindent
Since the deviation in headway indicates the possibility of extreme scenarios such as rear-end collisions or CACC disengagement,  the  detector output $r_i$ is defined as 
\begin{align}\label{r}
    r_i =  h_i - \hat{h}_i  =  N e_i,
\end{align}  where $\hat{h}_i$ is the estimated headway and $N = \begin{bmatrix}
    1 & 0 & 0 & 0
\end{bmatrix}\in {R}^{1\times 4} $. We note here that  \eqref{err} indicates that the proposed DS \eqref{chihat} is impacted by the preceding vehicle's error $e_{i-1}$.  Therefore, we define the generated \textit{V2X DS residual}, $r_{c, i}$ as
\begin{equation}\label{rc}
    r_{c, i} =  \max{\left(0,  \left[ \lVert r_i \rVert_2^2 - \lVert e_{i-1} \rVert_2^2 \right]\right)}.
\end{equation}
\subsection{Detector gain design}\label{v2x ds degin}
Table \ref{tab:th1_var} enlists the definitions of necessary constants and matrices to state the LMI criteria for the detector gain matrix $\mathfrak{L}_\alpha$ to ensure the desired performance in \mbox{Theorem \ref{th1}}.  

\begingroup   
\setlength{\tabcolsep}{10pt} 
\renewcommand{\arraystretch}{1.7} 
\begin{table}[h!]
    \centering
\begin{tabular}{|*{2}{c|}}
\hline
  \multicolumn{2}{|c|}{Constants and matrices utilized in Theorem \ref{th1}}\\
  \hline
  ${\gamma_1}_i =  - \min\limits_\alpha  \frac{\lambda_{\min[\Lambda_{1_{\alpha i}}]}}{\lambda_{\min[N^T N]}}$ & $\rho_{2_i} = \frac{\rho_{4_i}}{\rho_{3_i}}$  \\
  \hline
  $\mathfrak{N}  = N^T{P_\alpha} N$ & $\mathcal{W}_{\alpha i} =  \begin{bmatrix}
             {P_\alpha} N \mathcal{D}_{\theta \alpha} &  {P_\alpha} N  
        \end{bmatrix}$ \\
  \hline
  \multicolumn{2}{|c|}{$F_{\alpha i} =    \begin{bmatrix}
             {P_\alpha} N \mathcal{D}_{\theta \alpha} & {P_\alpha} N  E_{\alpha i} & {P_\alpha} N  E_{\theta i} & {P_\alpha} N  E_{\theta \alpha i}
        \end{bmatrix}$} \\
  \hline
  \multicolumn{2}{|c|}{$ \Lambda_{1_{\alpha i}} =  {\left(\mathcal{A}_{\alpha} - \mathfrak{L}_\alpha \mathcal{C} \right)}^T\mathfrak{N}    + \mathfrak{N}    {\left(\mathcal{A}_{\alpha} - \mathfrak{L}_\alpha \mathcal{C} \right)}$}\\
  \hline
  \multicolumn{2}{|c|}{$ \Lambda_{2_{\alpha i}} = \Lambda_{1_{\alpha i}} + \frac{1}{{\beta_1}_i} \mathfrak{N}   \mathcal{D}_{\theta \alpha} \mathcal{D}_{\theta \alpha}^T \mathfrak{N}  + \frac{1}{{\beta_2}_i} \mathfrak{N}   E_{\alpha i} E_{\alpha i}^T \mathfrak{N} $}\\
  \multicolumn{2}{|c|}{\hspace{5mm}$+ \frac{1}{{\beta_3}_i} \mathfrak{N}  E_{\theta i} E_{\theta i}^T \mathfrak{N}  + \frac{1}{{\beta_4}_i} \mathfrak{N}   E_{\theta \alpha i}  E_{\theta \alpha i}^T \mathfrak{N}  + \frac{1}{{\beta_5}_i} \mathfrak{N}  \mathfrak{N}$}\\
  \hline
  \multicolumn{2}{|c|}{$\mathcal{G}_1 = diag[I,-I -{\rho_{1_i}} I_4]$}\\
  \hline
  \multicolumn{2}{|c|}{$ \mathcal{G}_2  = diag[{\rho_{3_i}} I, -{\rho_{3_i}} I, -{\rho_{4_i}} I_4, -{\rho_{4_i}} I_4, -{\rho_{4_i}} I_4]$}\\
  \hline
  $ \mathcal{B}_{1_{\alpha i}} = \begin{bsmallmatrix}
        -{\gamma_1}_i I & \mathcal{W}_{\alpha i}  \\
       \mathcal{W}_{\alpha i}^T  & \textbf{0} 
        \end{bsmallmatrix}$ & $\mathcal{B}_{2_{\alpha i}} = \begin{bsmallmatrix}
            {\gamma_1}_i I & F_{\alpha i}\\
            F_{\alpha i}^T  & \textbf{0}
        \end{bsmallmatrix}$\\
        \hline
\end{tabular}
\vspace{1mm}
\caption{List of required constants and matrices}
    \label{tab:th1_var}
\end{table}
\endgroup

\begin{thmm}\label{th1}
    Let us  consider the vehicle platoon dynamics \eqref{zeta} and the V2X detector dynamics \eqref{chihat}. The $i^{th}$ detector  is considered  DSS, robust to uncertainty, and sensitive to cyberattacks for every mode $\alpha$, if there exists a set of positive symmetric matrices $\mathcal{P} := \{ {P_\alpha} : \alpha \in \mathcal{M} \}$ and constants ${\rho_j}_i,  {\beta_j}_i,{\gamma_1}_i > 0, \forall j \in \{1,2,3,4 \}$ and $0\leqslant {\mu}_i \leqslant 1$, such that the constrained multi-objective optimization problem posed below has a feasible solution.
        \begin{align}
        \max\limits_{\mathfrak{L}_\alpha} \big[{\mu}_i (-{\rho_{1_i}}) &+ (1-{\mu}_i) {\rho_{2_i}} \big],\hspace{2mm} \forall i \in \{1,2,\cdots, n\}; \label{opt}\\
        \text{Subjected to } &\Lambda_{2_{\alpha i}} < 0, \hspace{11mm} \text{(DSS criteria)}  \label{iss_con} \\ 
        &\mathcal{G}_1 + \mathcal{B}_{1_{\alpha i}} \leqslant 0,\hspace{3mm} \text{(Robustness criteria)}  \label{rob_con} \\
        &\mathcal{G}_2-\mathcal{B}_{2_{\alpha i}} \geqslant 0, \hspace{4mm} \text{(Sensitivity criteria)}  \label{sen_con}
    \end{align}
      where $\Lambda_{2_{\alpha i}},\mathcal{G}_1,\mathcal{B}_{1_{\alpha i}},\mathcal{G}_2, \mathcal{B}_{2_{\alpha i}}$ are defined in Table \ref{tab:th1_var}.
   
\end{thmm}

\begin{proof}
    The proof is shown in the appendix.
\end{proof}

\begin{remm}
    The DSS criteria ensure that the V2X DS residual $r_{c, i}(t)$ exponentially converges to zero under no cyberattacks and no uncertainties, i.\,e., $f_{\theta i}, f_{\alpha i}, f_{\theta \alpha i} , w_i  = 0$. This proves the nominal exponential stability criteria for the V2X DS. 
\end{remm}
Hence, from Theorem \ref{th1} we have obtained the LMI conditions for detector gains to ensure the  performance benchmarks for the V2X DS.  Next, we will define the V2X detector threshold to make an attack decision.

\subsection{V2X attack decision and V2I IS activator} \label{v2x threshold}
Each V2X DS residual generated from the $n$ number of detectors is compared with a pre-defined threshold ${J}_{DS}$ to make an attack detection. If the generated V2X DS residual  $r_{c, i}$ \eqref{rc} crosses the threshold ${J}_{DS}$, then a V2X attack decision is made and the second phase which is the V2I  IS is activated. The threshold ${J}_{DS}$ can be defined by observing the probability of the false alarm $P_{D}$  corresponding to that specific threshold value, i.\,e., the probability of the residual  $r_{c, i}$ crossing the threshold ${J}_{DS}$  under nominal operation. First, we generate the V2X DS residual data $r_{w}$  under no-attack and in presence of uncertainties scenarios ($\Delta_\alpha , \Delta_\theta = 0, w_i \neq 0 $) from Monte-Carlo simulations. In practical application, we can test-run the vehicle platoon under different driving conditions such as urban, suburban, and highway routes at diverse time intervals within a day to collect the necessary V2X DS residual data $r_{w}$. Then,  we can obtain the false alarm probability $P_{D}$ as: 
\vspace{-1.87mm}
\begin{align}\nonumber
    P_{D} & = P\left( r_{c, i} > {J}_{DS} | \Delta_\alpha = 0 \, \& \, \Delta_\theta = 0 \right)  = \int_{{J}_{DS}}^{\infty}\!\!\! P(r_{w}) \,dr_{w},
\end{align}
$r_{c, i}  \geqslant \mathcal{J}_{DS}$ implies V2X attack flag is generated for the $i^{th}$ vehicle and simultaneous activation of V2I IS, otherwise there is no attack.

\section{V2I Isolation Scheme} \label{ie}
Upon V2X attack flag generation and  V2I  IS activation, it is confirmed that single or multiple vehicles are under cyberattacks that can show up in vehicle dynamics as 3 cases:
\begin{description}

    \item[Case 1.] \textit{only V2V attack:} presence of \, $ \Tilde{\theta}_{i} \Tilde{u}_{i-1}$ \, ($E_{\theta i} f_{\theta i}$) term, 
    \item[Case 2.] \textit{only V2I attack:} presence of \, $ \Tilde{\alpha}_{i}$ \, ($E_{\alpha i} f_{\alpha i}$) term, and 
    \item[Case 3.] \textit{simultaneous V2V-V2I attacks:} presence of both  $ \Tilde{\theta}_{i} \Tilde{u}_{i-1}$ and \, $ \Tilde{\alpha}_{i}$ \,  ($E_{\theta \alpha i} f_{\theta \alpha i}$) terms.
\end{description}
Now our objective is to isolate the  V2I attacks, i.e., presence of $ \Tilde{\alpha}_{i}$ in the event of Cases 2 and 3 with residual-based observers.
To achieve this objective, the V2I IS acquires and utilizes  the data transmitted through the V2V network channel. Thus, if the compromised $i^{th}$ vehicle receives $\Tilde{\theta}_{i} \Tilde{u}_{i-1}$, the corresponding V2I IS model contains the same term. Hence, the incorporation of this V2V data essentially equips the V2I IS scheme with the ability to match both the nominal and corrupt V2V communication scenarios.
Therefore, in the event of Case 1, the V2I IS residual never crosses the threshold.

 Contrarily, in Cases 2 and 3, the V2I IS dynamics still run with expected mode $\alpha$, while the compromised vehicle dynamics run with the compromised mode $\Tilde{\alpha}_{i}$. This mismatch eventually drives the V2I IS residual to cross the threshold and an isolation flag is generated indicating the presence of V2I attack.  Mathematically, the V2V information is given through $ f_{\theta i} $, which is available to $i$-{th} V2I IS observer, while $f_{\alpha i},  f_{\theta \alpha i} $ are still unknown to the observer. 
Thus, the {  V2I isolator dynamics }is modeled following the system model \eqref{zeta} with the admitted V2V communication as 
\begin{align}\label{chihat2}
    \dot{\hat{\xi}}_i & =  \mathcal{A}_\alpha \hat{\xi}_i  + \mathcal{D}_{\theta \alpha} \hat{\xi}_{i-1} - \mathcal{H}_{\alpha} s_i +  E_{\theta i} f_{\theta i} +\mathfrak{M}_\alpha (y_i-\hat{Y}_i),\nonumber \\  \hat{Y}_i & = \mathcal{C} \hat{\xi}_i, \quad \quad t \geqslant t_0; \quad \quad \hat{\xi}_i(t_0) = \hat{\zeta}_i (t_0).
 \end{align}
Here $\hat{\xi}_i$, $\hat{Y}_i$, and $\mathfrak{M}_\alpha$ are respectively the isolator state, output, and gain matrix and $t_0$ is the time of V2X attack detection.  
  We note here, $\hat{\xi}_{i-1} (t) = \hat{\zeta}_{i-1} (t)$ if the $(i-1)^{th}$ isolator is not activated.
Now, the error between the system and the isolator dynamics is defined as $d_i := \zeta_i - \hat{\xi}_i$. 
From \eqref{zeta} and \eqref{chihat2} and with $\mathcal{A}_{\mathfrak{M}_\alpha} = \mathcal{A}_{\alpha} - \mathfrak{M}_\alpha \mathcal{C} $, the error dynamics for the IS:
\begin{align}\label{err2}
    \dot{d}_i = \mathcal{A}_{\mathfrak{M}_\alpha} d_i  +\mathcal{D}_{\theta \alpha} d_{i-1}  +  E_{\alpha i} f_{\alpha i}   + E_{\theta \alpha i} f_{\theta \alpha i}  +w_i.
\end{align}
Now, we denote the generated residual from these V2I IS observers as \textit{V2I IS residual} $\psi_i = h_i - \hat{H}_i = N d_i.$ Here, $\hat{H}_i$ is estimated headway under V2V attack. From \eqref{err2}, we
can conclude that  the generated V2I IS residual $| \psi_i |$ will exponentially  converge to zero during Case 1. However, during Cases 2 and 3,  the generated V2I IS residual $| \psi_i |$ will cross the threshold $\mathcal{J}_{IS}$ to create V2I attack flag. Furthermore, similar to the V2X  DS, the V2I  IS   gain matrix $\mathfrak{M}_\alpha$, and the pre-defined threshold $\mathcal{J}_{IS}$ can be designed by following  steps described in Section \ref{v2x ds degin}-\ref{v2x threshold}. We point out here that for both the V2X DS and the V2I IS, the individual residual comparison enables the proposed algorithm to \textit{identify the individual compromised vehicles} and thus, isolate the source of attacks.

 Once the detection-isolation decision has been made,  the system administrator can use an uncompromised channel or application to let the vehicles in the platoon know. At this stage, the platoon can disengage from standard driving and activate emergency ``limp mode home" \cite{dehghan2016sensorless}.

\section{Simulation Results} \label{4}

In this section, we present the simulation results 
for a CAV platoon of 12 vehicles with one leader vehicle and 11 follower CAVs. Furthermore, we consider two different driving environments for the platoon, namely, highway driving and urban driving. To simulate  the highway driving scenario, we extracted data for vehicle Id 1067 from \href{https://ops.fhwa.dot.gov/trafficanalysistools/ngsim.htm}{NGSIM} data and for the urban driving scenario, we utilized the \href{https://www.epa.gov/emission-standards-reference-guide/epa-urban-dynamometer-driving-schedule-udds}{UDDS} data. 

Next, for the follower vehicles of the platoon, we choose the engine dynamics time constant to be $\Sigma = 0.1s$. For  nominal highway driving the vehicle controllers have gain values $k_1 = \begin{bmatrix}
    0.2 & 1 & 2
\end{bmatrix}$ and the desired time-headway $\mathcal{T}_1 = 0.5s$. Similarly, for nominal urban driving, the vehicle controllers have gain values $k_2 = \begin{bmatrix}
    0.2 & 2 & -0.6
\end{bmatrix}$ and the desired time-headway $\mathcal{T}_2 = 1.2s$ which is larger than that of highway driving due to the frequent stop-and-go scenario.
The detector gains, $\mathfrak{L}_1 $ and isolator gains, $\mathfrak{M}_1 $ for highway driving environment we choose $\mathfrak{L}_1 = \mathfrak{M}_1 =  \begin{bmatrix}
    3 & 0.2 & 5 & 0.9 \\
    3 & -0.2 & -0.5 & 0.8
\end{bmatrix}$ 
such that they ensure performance benchmarks. Similarly, for the urban driving environment, are chosen to be $\mathfrak{L}_2 = \mathfrak{M}_2 =  \begin{bmatrix}
    0.6 & 0.2 & 0.8 & 0.2 \\
    1.2 & 0.2 & 0.8 & 0.2
\end{bmatrix}$. 
Additionally, the system uncertainty is considered to be Gaussian i.e., $w_i \sim \mathcal{N}(0,1)[1,1]$ with maximum $\pm 5\%$ of velocity. The values of the thresholds $\mathcal{J}_{DS}$, and $\mathcal{J}_{IS}$ are chosen to be 0.25 and 0.75 respectively. Now, we craft three attack scenarios to capture the various CAV attack complexities and variabilities such as single or multiple vehicle disruptions, different driving conditions, potentially disastrous consequences,  and all possible cases of network corruption, i.e. only V2V, only V2I, and V2X attacks.

\subsection{Replay attack on single vehicle through V2V network}
\begin{figure}
    \centering
    \includegraphics[trim = 6mm 0mm 0mm 0mm, clip,  width=0.8\linewidth]{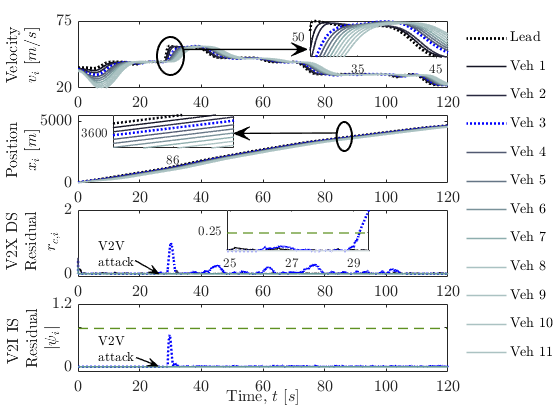}   
    \caption{Under V2V cyberattack, the plot shows (top)  the velocity, (second) the position, (third) the V2X DS residual, and (last) the V2I IS residual.}
    \label{fig:v2v}
\end{figure}
First, we simulate the only V2V attack scenario ($f_{\theta i} \neq 0$, $f_{\alpha i}, f_{\theta \alpha} i = 0$)  on the vehicle platoon running in a highway scenario. The V2V attack impacts the $3^{rd}$ vehicle, i.e., the V2V communication between the $2^{nd}$ vehicle and the $3^{rd}$ vehicle has been compromised. The adversary first records $u_2$ signal from $18^{th} s$ to $25^{th} s$ and then starts to repeatedly send this $7s$-long signal  to the $3^{rd} $ vehicle from $25^{th} s$.   Thus, the $3^{rd}$ vehicle effectively tracks the wrong desired acceleration leading to string instability \cite{ploeg2011design}.   
This is evident from the top two plots of \mbox{Fig. \ref{fig:v2v}} (especially the zoomed inset plots) showing the position and velocity of the vehicles in the platoon.   
Now, the V2X detector can detect the attack on this platoon within 5s as the $3^{rd} $ vehicle V2X DS residual $r_{c, 3}$  crosses the threshold as shown in the third plot of \mbox{Fig. \ref{fig:v2v}}.
After a V2X detection, the $3^{rd}$ V2I isolator is activated. However, since the platoon is under a V2V attack, the V2I IS residual $|\psi_3|$ remains under the threshold for all time (last plot in \mbox{Fig. \ref{fig:v2v}}).

\subsection{FDI attack on single vehicle through V2I network}
\begin{figure}
    \centering
    \includegraphics[trim = 3mm 0mm 0mm 0mm, clip,  width=0.8\linewidth]{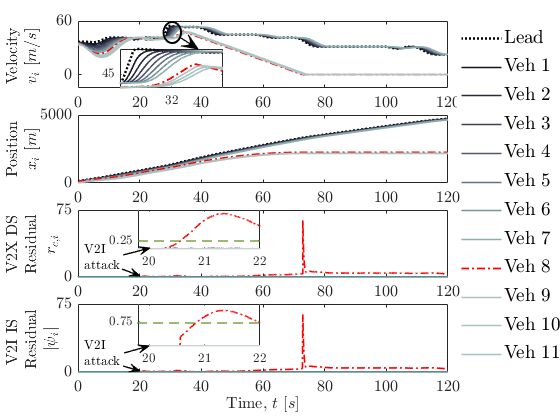}   
    \caption{Under V2I cyberattack, the plot shows (top)  the velocity, (second) the position, (third) the V2X DS residual, and (last) the V2I IS residual.}
    \label{fig:v2i}
\end{figure}
\vspace{-2mm}Here we simulate the platoon running in a highway scenario under the only V2I attack scenario ($f_{\alpha i} \neq 0$, $f_{\theta i}, f_{\theta \alpha i} = 0$) which impacts the $8^{th}$ vehicle from $t = 20 s$.
In this FDI attack,  the controller mode signal is $\alpha$ altered such that the $8^{th}$ vehicle effectively runs with the controller gains and the desired headway for the urban driving environment. This causes the vehicle to lag behind the $7^{th}$ vehicle further with time and ultimately CACC disengagement due to the vast distance from the preceding vehicle. The rest of the vehicles that are following the $8^{th}$ vehicle are also disengaged. This is evident from the top two plots of \mbox{Fig. \ref{fig:v2i}}  showing the position and velocity of the vehicles in the platoon, which is also expected from previous experimental studies \cite{gunter2020commercially}. Now, the V2X residual of the vehicle 8  $r_{c, 8}$ crosses the threshold within 2s of attack occurrence to accurately detect the attack as shown in the third plot of \mbox{Fig. \ref{fig:v2i}}.
After a V2X detection, the $8^{th}$ V2I isolator is activated. Furthermore, as the platoon is under V2I attack, the V2I IS residual $|\psi_8|$ crosses the threshold  confirming the presence of the V2I attack  within 2s of activation. This is shown in the last plot in \mbox{Fig. \ref{fig:v2i}}.

\subsection{DDoS \& FDI attacks on multiple vehicles through  V2V-V2I networks}
\begin{figure}
    \centering
    \includegraphics[trim = 5mm 0mm 0mm 0mm, clip,  width=0.8\linewidth]{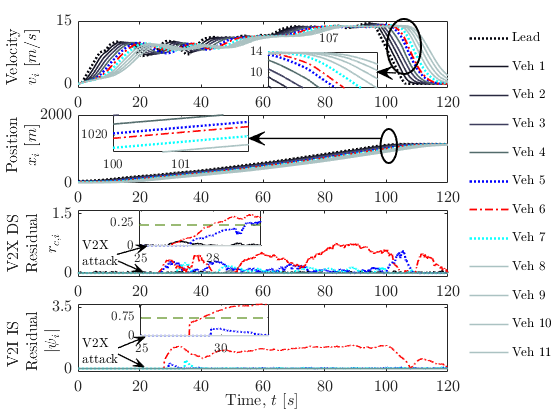}   
    \caption{Under V2X cyberattack, the plot shows (top)  the velocity, (second) the position, (third) the V2X DS residual, and (last) the V2I IS residual.}
    \label{fig:v2x}
\end{figure}
Lastly, we simulate the vehicle platoon under  the simultaneous V2V-V2I attack scenario ($f_{\theta i}, f_{\alpha i}, f_{\theta \alpha i} \neq 0$)   running in an urban driving scenario. Firstly, the adversary injects the DDoS attack on the V2V network of the $5^{th} $ to $7^{th} $ follower vehicles to disconnect them. Next, the adversary injects an FDI attack on the V2I network of the $6^{th} $ vehicle. The V2V and V2I attacks are injected at $t = 25 s$  and these lead to string instability and a rear-end collision scenario which is clearly evident in the first two plots of \mbox{Fig. \ref{fig:v2x}} showing the position and velocity of the vehicles in the platoon. Particularly, the zoomed portion of the velocity and position trajectory indicates the high probability of rear-end collisions. Furthermore, the V2X detector can detect the presence of the cyberattacks within 4s and activates corresponding isolators as shown in \mbox{Fig. \ref{fig:v2x}}. Then the $6^{th}$ V2I IS residual $|\psi_6|$  crosses the threshold and generates the V2I attack flag for this vehicle within 2s of activation.  On the other hand,  $ |\psi_5|,\, \& \,|\psi_7|$ remain below the threshold since only the V2V network is compromised for these vehicles.  This is shown in the last plot of \mbox{Fig. \ref{fig:v2x}}.

\subsection{Attack sensitivity and statistical significance}\label{stat section}
Reliable V2X attack detection and V2I attack isolation under several presently prevalent cyberattacks on the real-world NGSIM and UDDS data demonstrate the efficacy of the proposed algorithm in a practical setting. Next, to further evaluate the effectiveness of the algorithm at the chosen threshold levels under diverse attack intensities we perform a Monte Carlo simulation of 720 test runs on the algorithm and obtain the receiver operating characteristics or ROC curve. The algorithm exhibits an optimal balance between the attack sensitivity and the robustness (detection probability and false alarm probability) at the chosen thresholds $\mathcal{J}_{DS}$ and $\mathcal{J}_{IS}$ as shown in the left plot of the \mbox{Fig \ref{fig:stat}}. Furthermore, the right plot in \mbox{Fig \ref{fig:stat}} exhibits that on average, the algorithm can reliably detect the presence of the cyberattacks within 11 seconds of occurrence. However, low-impact attacks may remain undetected by the proposed algorithm leading to the 93$\%$ detection probability at the chosen threshold $\mathcal{J}_{DS}$.

\begin{figure}[h!]
    \centering
    \includegraphics[width=0.9\linewidth]{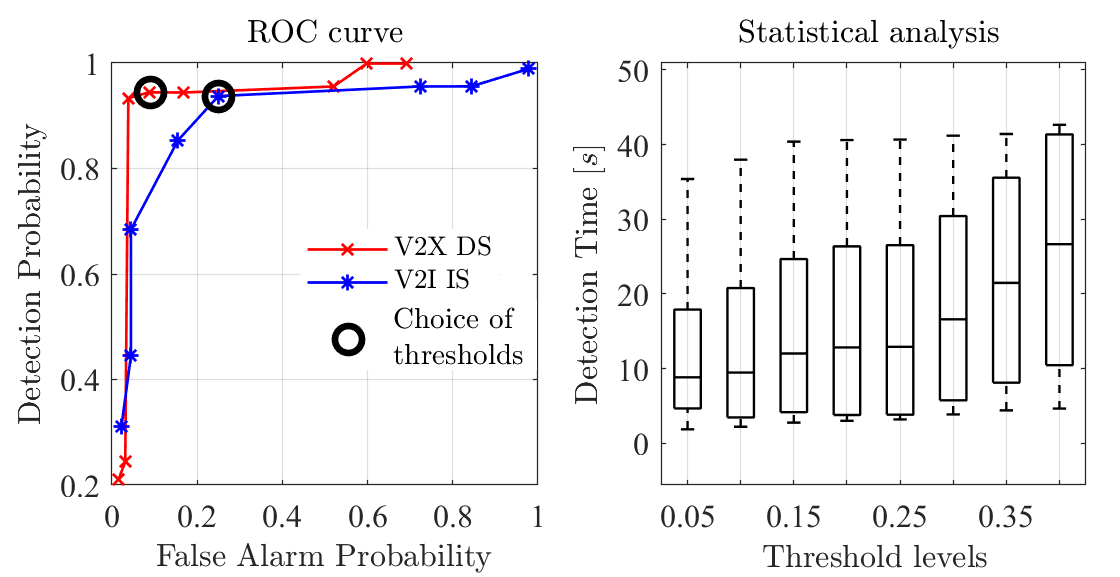}
    \caption{ Plots exhibit algorithm performance under diverse attack intensities at different threshold settings: (left) Receiver Operating Characteristics or ROC curve and (right) Statistical analysis of detection time.}
    \label{fig:stat}
\end{figure}

\section{Conclusion} \label{5}
In summary, we proposed a two-phase VANET cyberattack detection-isolation algorithm to first detect cyberattacks and thereafter isolate the presence of V2I attacks for a CAV platoon.  We  considered multiple operational modes for the vehicle platoon in order to accommodate the changing driving environments of the platoon. Additionally, these model-based detectors and isolators  were designed to ensure DSS with optimal robustness and sensitivity guarantees. Furthermore, we demonstrated the effectiveness of our algorithm using simulation case studies for three attack scenarios: only V2V, only V2I, and simultaneous V2V-V2I. The CAV driving environments were simulated here using UDDS and NGSIM data along with common cyberattacks such as DDoS, FDI, and replay attacks. In these case studies, our proposed algorithm successfully detected vehicular cyberattacks and isolated vehicles under a V2I attack.  While additional experiments are needed in CAV testbeds before practical adoption  and further analysis may be required to reevaluate the algorithm performance, if and when a new attack pattern emerges in the future, such detection decisions will ultimately help traffic administrators to employ steps towards damage mitigation while V2I attack isolation will help to secure central infrastructure access.

\section*{APPENDIX}
    \begin{table}[h!]
    \centering
\begin{tabular}{|*{2}{c|}}
\hline
  \multicolumn{2}{|c|}{Constants utilized in DSS criteria analysis}\\
  \hline
  ${\beta_6}_i = \frac{{\beta_1}_i}{\lambda_{\min[N^TN]}}$ & ${\overline{\beta}_{F}}_i = \max \left[{\beta_2}_i, {\beta_3}_i , {\beta_4}_i \right]$  \\
  \hline
  $ {\gamma_F}_i = \frac{{\overline{\beta}_{F}}_i}{{\gamma_2}_i \lambda_{\min[{P_\alpha}]}} $ & $ \lVert F_i \rVert_2^2 = \lVert f_{\alpha i} \rVert_2^2 + \lVert f_{\theta i} \rVert_2^2 +\lVert f_{\theta \alpha i} \rVert_2^2 $\\
  \hline
  ${\gamma_P} = \frac{\lambda_{\max[{P_\alpha}]}}{\lambda_{\min[{P_\alpha}]}}$  & ${\gamma_2}_i =  - \min\limits_\alpha \frac{\lambda_{\min[\Lambda_{2_{\alpha i}}]}}{\lambda_{\min[\mathfrak{N}]}}$ \\
  \hline
  $ {\gamma_w}_i = \frac{{\beta_5}_i}{{\gamma_2}_i \lambda_{\min[{P_\alpha}]}} $ & ${\gamma}_i = \frac{{\beta_6}_i {\gamma_P}}{\left({\gamma_2}_i + {\gamma_2}_{i-1}\right)\lambda_{\min[{P_\alpha}]}}$ \\
  \hline
  \multicolumn{2}{|c|}{$\mathfrak{P}_i = \max\limits_{l\in \{1,\cdots,i\}}{[{\gamma}_l,{\gamma_{F}}_l,{\gamma_{w}}_l]}\prod\limits_{k = l+1}^i \frac{{\beta_6}_k}{{\gamma_2}_k\lambda_{\min[{P_\alpha}]}} $} \\
  \hline
\end{tabular}
\vspace{1mm}
\caption{List of required constants}
    \label{tab:DSS_var}
\end{table}
\vspace{-7.mm}
\noindent
\textbf{Proof of Theorem 1:}
 Let us  consider a family of Lyapunov  functions for the switching-mode detector dynamics \eqref{chihat} as:
\begin{equation}\label{v}
    \Gamma_{\alpha i} = r_i^T {P_\alpha} r_i = e_i^T N^T {P_\alpha} N e_i = e_i \mathfrak{N}    e_i ,
\end{equation} where ${P_\alpha} \in \mathcal{P}$. Differentiating the Lyapunov function with respect to time we obtain
\begin{align}
     \dot{\Gamma}_{\alpha i} =  & e_i^T \Lambda_{1_{\alpha i}} e_i + q \left( e_i, \mathfrak{N}   \mathcal{D}_{\theta \alpha}, e_{i-1} \right) \nonumber \\ & +  q \left( e_i, \mathfrak{N}    E_{\alpha i}, f_{\alpha i} \right)  
     + q \left( e_i, \mathfrak{N}    E_{\theta i}, f_{\theta i} \right) \nonumber \\ & + q \left( e_i, \mathfrak{N}    E_{\theta \alpha i}, f_{\theta \alpha i} \right)+ q \left( e_i, \mathfrak{N} ,   w_i \right).  \label{vdiss}
\end{align}


 \noindent
    Now, utilizing the Lyapunov function, we will prove the theorem in three parts:- DSS, robustness, and sensitivity.

    \textbf{DSS:} The proposed detectors \eqref{chihat} are considered DSS, if   for the vehicle platoon dynamics \eqref{zeta}  under cyberattacks and uncertainties  (i.\,e., $f_{\alpha i}, f_{\theta i}, f_{\theta \alpha i} , w_i \neq 0$) the generated V2X DS  residual stays bounded, i.e. $r_{c, i} < \infty, \forall i \in \{1,\cdots,n\}$ . Before presenting the analysis, we present the definition of the required constants in the Table below.
Now, for $f_{\alpha i}, f_{\theta i}, f_{\theta \alpha i} , w_i \neq 0$, we have the derivative of $\Gamma_{\alpha i}$ from \eqref{vdiss}. Then,  using the formula $X^T Y + Y^T X \leqslant \beta X^T X + \frac{1}{\beta}Y^T Y$, for $\beta > 0$, we can write \eqref{vdiss} as
\begin{align}
    \dot{\Gamma}_{\alpha_i} \leqslant & e_i^T \Lambda_{2_{\alpha i}} e_i  + {\beta_1}_i \lVert e_{i-1} \rVert_2^2 + {\beta_2}_i \lVert f_{\alpha i} \rVert_2^2 + {\beta_3}_i \lVert f_{\theta i} \rVert_2^2 \nonumber \\ &+ {\beta_4}_i \lVert f_{\theta \alpha i} \rVert_2^2 + {\beta_5}_i \lVert w_i \rVert_2^2, \label{r_i}
\end{align} where $\Lambda_{2_{\alpha i}}$ is defined in Table \ref{tab:th1_var}.
Next, from \eqref{r}, we can write $\lambda_{\min[N^TN]} \lVert e_{i-1} \rVert_2^2  \leqslant \lVert r_{i-1} \rVert_2^2 $. Furthermore, we obtain  $e_i^T \Lambda_2{_{\alpha i}} e_i \leqslant -\lambda_{\min[\Lambda_{2_{\alpha i}}]} \lVert e_i\rVert_2^2$, for $\Lambda_{2_{\alpha i}} < 0$. Similarly, from \eqref{v} we can deduce  $\Gamma_{\alpha i} \geqslant \lambda_{\min[\mathfrak{N}]}\lVert e_i\rVert_2^2$. Thus, using ${\gamma_2}_i$,  ${\beta_6}_i$, ${\overline{\beta}_{F}}_i$, and $\lVert F_i \rVert_2^2$ defined in Table \ref{tab:DSS_var}, \eqref{r_i} can be rewritten as
\begin{align}
    \dot{\Gamma}_{\alpha_i} \leqslant \!\! - {\gamma_2}_i \Gamma_{\alpha i}  + {\beta_6}_i \lVert r_{i-1} \rVert_2^2 + {\overline{\beta}_{F}}_i \lVert F_{ i} \rVert_2^2 +  {\beta_5}_i \lVert w_i \rVert_2^2. \label{beta6}
\end{align}
The presence of $r_{i-1}$ in \eqref{beta6}, shows that the cumulative impact of all the preceding vehicles, $\{0,1,\cdots i-1\}$ is embedded in the $i^{th}$ detector. Thus, we proceed to prove this recursively. 
First, we prove for $i =1$. We note here, for the leader vehicle i.e., $i = 0$, the full state is known to V2X DS from GNSS data (since  $u_0 = a_0$), and thus the error $e_0 = \zeta_0 - \hat{\zeta}_0 = 0 $. Therefore, from  \eqref{r} we obtain $\lVert r_0 \rVert^2_2 \leqslant \lambda_{\max[N^TN]} \lVert e_{0} \rVert_2^2 = 0$. Hence, for $i = 1$, applying Gronwall's inequality on \eqref{beta6}
\begin{align}\label{vlong}
    \Gamma_{\alpha 1}(t) \leqslant &\Gamma_{\alpha 1}(0) e^{-{\gamma_2}_1 t} + \frac{{\overline{\beta}_{F}}_1}{{\gamma_2}_1}  \lVert F_{ 1} \rVert^2_{\mathcal{L}_2} +  \frac{{\beta_5}_1}{{\gamma_2}_1}  \lVert w_1 \rVert^2_{\mathcal{L}_2}.
\end{align}
Then, using \eqref{v}, we can write $\lambda_{\min[{P_\alpha}]} \lVert r_1 (t) \rVert_2^2 \leqslant  \Gamma_{\alpha 1} (t)$ and $\Gamma_{\alpha 1}(0) \leqslant  \lambda_{\max[{P_\alpha}]}  \lVert r_1(0) \rVert_2^2$. Thus, we utilize the variables defined in Table \ref{tab:DSS_var} for $i =1$ such that we re-write \eqref{vlong}  as
\begin{align}
    \lVert r_1 (t) \rVert_2^2 \leqslant & {\gamma_P} \lVert r_1(0) \lVert_2^2 e^{-{\gamma_2}_1 t}  +{\gamma_{F}}_1 \lVert F_{ 1} \rVert^2_{\mathcal{L}_2} + {\gamma_{w}}_1 \lVert w_1 \rVert^2_{\mathcal{L}_2}. \nonumber
\end{align}
With this, we prove that detector 1 is DSS. Henceforth, and using \eqref{beta6} for $i = 2$ we obtain
\begin{align}
    \dot{\Gamma}_{\alpha_2} &\leqslant  - {\gamma_2}_2 \Gamma_{\alpha 2}   + {\overline{\beta}_{F}}_2 \lVert F_{ 2} \rVert_2^2  + {\beta_5}_2 \lVert w_2 \rVert_2^2  + {\beta_6}_2 \nonumber \\ &\left[ {\gamma_P} \lVert r_1(0) \lVert_2^2 e^{-{\gamma_2}_1 t}  + {\gamma_{F}}_1 \lVert F_{1} \rVert^2_{\mathcal{L}_2} + {\gamma_{w}}_1 \lVert w_1 \rVert^2_{\mathcal{L}_2}\right] \label{beta6_2}.
\end{align}
Following similar steps after applying Gronwall's inequality and using  ${\gamma}_2$ from Table \ref{tab:DSS_var} we obtain
\begin{align}
    &\lVert r_2 (t) \rVert_2^2 \leqslant  {\gamma_P} \lVert r_2(0) \lVert_2^2 e^{-{\gamma_2}_2 t}  +  {\gamma}_2 \lVert r_1(0) \rVert^2_{\mathcal{L}_2} +  {\gamma_{F}}_2 \lVert F_{2} \rVert^2_{\mathcal{L}_2} \nonumber \\ & +    {\gamma_{w}}_2 \lVert w_2 \rVert^2_{\mathcal{L}_2} +   \frac{{\beta_6}_2}{{\gamma_2}_2\lambda_{\min[{P_\alpha}]}} \left[{\gamma_{F}}_1 \lVert F{1} \rVert^2_{\mathcal{L}_2}   +  {\gamma_{w}}_1 \lVert w_1 \rVert^2_{\mathcal{L}_2}\right]. \nonumber
\end{align}
Thus, for  $\mathfrak{P}_i$ from Table \ref{tab:DSS_var}, iteratively  it can be shown that 
 \begin{align}\label{r_dss}
      \lVert r_i (t) \rVert_2^2 &\leqslant  {\gamma_P} \lVert r_i(0) \lVert_2^2 e^{-{\gamma_2}_i t}  + \nonumber \\ &\mathfrak{P}_i \sum\limits_{l=1}^i \left[ \lVert r_{l-1}(0) \rVert^2_{\mathcal{L}_2} +   \lVert F_{l} \rVert^2_{\mathcal{L}_2}  +   \lVert w_l \rVert^2_{\mathcal{L}_2}\right]. 
 \end{align}
With this, we prove DSS criteria for $r_i$. Furthermore, this implies that $r_{i-1}$ and subsequently $e_{i-1}$ is also bounded under the presence of cyberattack and uncertainty since from \eqref{r} we can write $\lambda_{\min[N^TN]} \lVert e_{i-1} \rVert_2^2  \leqslant \lVert r_{i-1} \rVert_2^2$. If $r_i$ is greater than $e_{i-1}$ then $\lVert r_i \rVert_2^2 - \lVert e_{i-1} \rVert_2^2$ is positive and less than $r_i$. This implies that under this condition $r_{c, i}$ is less than equal to the right-hand-side (RHS) of \eqref{r_dss}. Moreover, if $r_i$ is less than $e_{i-1}$, $r_{c, i} =0$. However, the RHS of \eqref{r_dss} is always non-negative, which implies for the $r_i$ less than $e_{i-1}$ case, then $r_{c, i}$ is still less than equal to the RHS of \eqref{r_dss}. This yields
 \begin{align}
      r_{c, i} &\leqslant  {\gamma_P} \lVert r_i(0) \lVert_2^2 e^{-{\gamma_2}_i t}  + \mathfrak{P}_i \sum\limits_{l=1}^i \Big[ \lVert r_{l-1}(0) \rVert^2_{\mathcal{L}_2} \nonumber \\ &+   \lVert f_{\alpha l} \rVert^2_{\mathcal{L}_2}  +  \lVert f_{\theta l} \rVert^2_{\mathcal{L}_2} +   \lVert f_{\theta \alpha l} \rVert^2_{\mathcal{L}_2} +  \lVert w_l \rVert^2_{\mathcal{L}_2}\Big]. 
 \end{align}
This ensures the DSS criteria for the proposed detector \eqref{chihat}.

    \textbf{Robustness:} The $i^{th}$ detector is considered robust towards uncertainties $(w_i \neq 0)$, if under no cyberattack ($f_{\alpha i}, f_{\theta i},  f_{\theta \alpha i} = 0)$ the V2X DS  residual $r_{c, i}(t)$ does not cross a predefined threshold value. Thus, the detector is robust if there exists RF, $\rho_{1_i} > 0$ such that $\forall i \in \{1,\cdots,n\}$
\begin{align}\label{rob_rc}
    \int_0^\infty r_{c, i} dt \leqslant \rho_{1_i} \int_0^\infty \lVert w_i \rVert^2_2 dt + \epsilon_{1}. 
\end{align}
Furthermore, the non-negative RHS of the inequality in \eqref{rob_rc} implies that proving $\lVert r_i \rVert_2^2 - \lVert e_{i-1} \rVert_2^2 $ is smaller than the RHS will prove each component  of $r_{c, i}$ \eqref{rc}  is individually less than the RHS. Therefore, it is sufficient to prove
\begin{align}\label{rob}
    \int_0^\infty \left[ \lVert r_i \rVert_2^2 - \lVert e_{i-1} \rVert_2^2 \right] dt \leqslant \rho_{1_i} \int_0^\infty \lVert w_i \rVert^2_2 dt + \epsilon_{1}, 
\end{align} to ensure robustness. First, let us define a vector $\chi_{1 i} = \begin{bmatrix}
        r_i^T & e_{i-1}^T & w_i^T
    \end{bmatrix}^T$. Now, the robustness criteria from \eqref{rob} can be re-written using the definition of $\mathcal{G}_1,$ from \mbox{Table \ref{tab:th1_var}} as $\int_0^\infty \chi_{1 i}^T  \mathcal{G}_1 \chi_{1 i} dt -\epsilon_{1} \leqslant 0$. Next, let us consider the same Lyapunov function from \eqref{v}. Since  $f_{\alpha i},  f_{\theta i}, f_{\theta \alpha i}= 0, w_i \neq 0$ for this criteria, the derivative of ${\Gamma}_{\alpha i}$ is modified to $\dot{\Gamma}_{\alpha i} \leqslant \chi_{1 i}^T \mathcal{B}_{1_{\alpha i}} \chi_{1 i}$ using \eqref{vdiss}, where $\mathcal{B}_{1_{\alpha i}},$ is defined in \mbox{Table \ref{tab:th1_var}}. Thus,
    integrating this inequality we  obtain
    \begin{align}\label{l0}
        \Gamma_{\alpha i}(\infty) - \Gamma_{\alpha i}(0) - \int_0^\infty \chi_{1 i}^T \mathcal{B}_{1_{\alpha i}} \chi_{1 i} dt \leqslant 0.
    \end{align}
    Since the left-hand-side (LHS) of the \eqref{l0} is non-positive, if $\int_0^\infty \chi_{1 i}^T  \mathcal{G}_1 \chi_{1 i}dt -\epsilon_{1}$ can be proved to be smaller than the LHS of the \eqref{l0}, then $\int_0^\infty \chi_{1 i}^T  \mathcal{G}_1 \chi_{1 i}dt -\epsilon_{1}$ will also be non-positive. Mathematically, it implies that we can prove
    \begin{align}
        \int_0^\infty \! \! \! \! \! \chi_{1 i}^T \mathcal{G}_1 \chi_{1 i} dt -\epsilon_{1} \leqslant \Gamma_{\alpha i}(\infty) - \Gamma_{\alpha i}(0) - \int_0^\infty \!\! \! \! \! \chi_{1 i}^T \mathcal{B}_{1_{\alpha i}} \chi_{1 i} dt, \nonumber
    \end{align} to prove \eqref{rob}. This condition can be further modified as
    \begin{align}\label{r2}
        \int_0^\infty \! \! \! \! \chi_{1 i}^T \left[\mathcal{G}_1+ \mathcal{B}_{1_{\alpha i}}\right] \chi_{1 i} dt \leqslant \Gamma_{\alpha i}(\infty) - \Gamma_{\alpha i}(0) + \epsilon_{1}.
    \end{align} Now, by definition $\Gamma_{\alpha i}(\infty) \geqslant 0$. Using \eqref{v}, we can write $\Gamma_{\alpha i}(0) = r_i^T (0) P_\alpha r_i (0) \leqslant \lambda_{\max[P_\alpha]} \max\limits_i \lVert r_i (0) \rVert_2^2$. Thus, we can deduce $\Gamma_{\alpha i}(0) - \epsilon_{1}\leqslant 0$ for $\epsilon_{1} \geqslant \lambda_{\max[P_\alpha]} \max\limits_i \lVert r_i (0) \rVert_2^2$. Therefore, we can write
    \begin{align}\label{v0g0}
        \Gamma_{\alpha i}(\infty) - \Gamma_{\alpha i}(0) + \epsilon_{1}\geqslant 0.
    \end{align} Hence, if we can guarantee 
    \begin{align}\label{r3}
        \int_0^\infty \! \! \! \! \chi_{1 i}^T \left[\mathcal{G}_1+\mathcal{B}_{1_{\alpha i}}\right] \chi_{1 i} dt \leqslant 0,
    \end{align} then it will be sufficient to prove \eqref{r2} using \eqref{v0g0}. \eqref{r3} thus provides us with the LMI criteria $\mathcal{G}_1 + \mathcal{B}_{1_{\alpha i}} \leqslant 0$ that proves \eqref{r2} which is equivalent to prove \eqref{rob}. Thus, we prove that  \eqref{rob_con} is a sufficient condition for guaranteeing robustness.

    \textbf{Sensitivity:}\, In the presence of cyberattacks and no uncertainties, a detector is considered sensitive towards cyberattack if the V2X DS residual $r_{c, i} (t)$ can cross a predefined threshold value regardless of the amplitude of the attack. Furthermore, from \eqref{rc}, we can write $r_{c, i} \geqslant \lVert r_i \rVert_2^2 - \lVert e_{i-1} \rVert_2^2$. Thus, for $f_{\alpha i}, f_{\theta i},  f_{\theta \alpha i} \neq 0$ and $w_i = 0$, the detector is sensitive if there exists SF, $\rho_{2_i} > 0$ and $\epsilon_{1}>0$, such that $\forall i \in \{1,\cdots,n\}$
\begin{align}\label{sproof}
    &\int_0^\infty\!\! \! \! \! r_{c, i} dt \geqslant \int_0^\infty\!\! \! \! \! \left[\lVert r_i \rVert_2^2 - \lVert e_{i-1} \rVert_2^2 \right] dt \geqslant \nonumber \\ & \rho_{2_i} \int_0^\infty \! \! \! \! \left[ \lVert f_{\alpha i} \rVert^2_2+ \lVert f_{\theta i} \rVert^2_2 + \lVert f_{\theta \alpha i} \rVert^2_2\right] dt - \epsilon_{1}.
\end{align}
Hence, we note that to prove the sensitivity condition it is adequate to prove the second inequality in \eqref{sproof}. 
Now, for design flexibility we substitute  $\rho_{2_i} = \frac{\rho_{4_i}}{\rho_{3_i}}$ in \eqref{sproof}  to get 
    \begin{align}\label{sk}
        \rho_{3_i}& \int_0^\infty \! \! \! \! \left[ \lVert r_i \rVert_2^2 - \lVert e_{i-1} \rVert_2^2 \right] dt \geqslant  \rho_{4_i} \int_0^\infty \! \! \! \! \!  \lVert F_{i} \rVert^2_2 dt - \epsilon_{2},
    \end{align} where $\epsilon_{2} = \rho_{3_i} \epsilon_{1}$, and $\lVert F_i \rVert_2^2$ is defined in Table \ref{tab:DSS_var}.
Now to prove this, let us define another vector $\chi_{2 i} = \begin{bmatrix}
        r_i^T & e_{i-1}^T & f_{\alpha i}^T & f_{\theta i}^T & f_{\theta \alpha i}^T
    \end{bmatrix}^T$.  Next, using the definition of $\mathcal{G}_2$ from \mbox{Table \ref{tab:th1_var}}, \eqref{sk} can be re-written as:
    \begin{align}
        \int_0^\infty \chi_{2 i}^T \mathcal{G}_2 \chi_{2 i} dt + \epsilon_{2}\geqslant 0.\label{s1}
    \end{align}  Then considering the same Lyapunov functions defined in \eqref{v} under  $f_{\alpha i}, f_{\theta i},  f_{\theta \alpha i} \neq 0$ and $w_i = 0$, the   derivative of ${\Gamma}_{\alpha i}$ from \eqref{vdiss} can be modified as $ \dot{\Gamma}_{\alpha i} \leqslant  \chi_{2 i}^T \mathcal{B}_{2_{\alpha i}} \chi_{2 i} $, where $\mathcal{B}_{2_{\alpha i}}$ is defined in \mbox{Table \ref{tab:th1_var}}. 
    We note here that to obtain this inequality we have utilized the fact, for ${\gamma_1}_i>0$, ${\gamma_1}_i r_i^T r_i \geqslant -{\gamma_1}_i r_i^T r_i$. Next, we integrate the inequality to get
    \begin{align}\label{s3}
        \int_0^\infty\!\! \! \! \! \chi_{2 i}^T \mathcal{B}_{2_{\alpha i}} \chi_{2 i} dt - \Gamma_{\alpha i}(\infty) + \Gamma_{\alpha i}(0) \geqslant 0.
    \end{align} Since LHS of \eqref{s3} is positive,  if  $\int_0^\infty \! \! \chi_{2 i}^T \mathcal{G}_2 \chi_{2 i} dt + \epsilon_{2}$ can be proved to be larger than the LHS of \eqref{s3}, then $\int_0^\infty \! \! \chi_{2 i}^T \mathcal{G}_2 \chi_{2 i} dt + \epsilon_{2}$ will also be positive. This implies that we must prove 
    \begin{align}
       \! \! \! \! \int_0^\infty\!\! \! \! \! \chi_{2 i}^T \mathcal{G}_2 \chi_{2 i} dt + \epsilon_{2}\geqslant \! \! \!\int_0^\infty \!\! \! \! \!\chi_{2 i}^T \mathcal{B}_{2_{\alpha i}} \chi_{2 i} dt - \Gamma_{\alpha i}(\infty) + \Gamma_{\alpha i}(0), \nonumber
    \end{align} to prove \eqref{s1}. This inequality condition can be rearranged as
    \begin{align}\label{s4}
        \int_0^\infty \!\! \! \! \!\chi_{2 i}^T \left[ \mathcal{G}_2 - \mathcal{B}_{2_{\alpha i}} \right] \chi_{2 i} dt \geqslant \Gamma_{\alpha i}(0) - \epsilon_{2} - \Gamma_{\alpha i}(\infty).
    \end{align} Furthermore, from \eqref{v0g0}, we can write  $\Gamma_{\alpha i}(0) - \epsilon_{2}- \Gamma_{\alpha i}(\infty) \leqslant 0$ for $\epsilon_{2}  \geqslant \lambda_{\max[P_\alpha]} \max\limits_i \lVert r_i (0) \rVert_2^2$. Consequently, we can deduce that when $\int_0^\infty \! \! \! \chi_{2 i}^T \left[ \mathcal{G}_2 - \mathcal{B}_{2_{\alpha i}} \right] \chi_{2 i} dt$ is positive, then the inequality \eqref{s4} will be satisfied. Thus, if we can guarantee  
    \begin{align}\label{s5}
        \int_0^\infty \chi_{2 i}^T \left[ \mathcal{G}_2 - \mathcal{B}_{2_{\alpha i}} \right] \chi_{2 i} dt \geqslant 0,
    \end{align} it will  prove \eqref{s4}  which in turn proves \eqref{sk}. Therefore, \eqref{s5} implies that the LMI condition proposed in \eqref{sen_con}, \mbox{i.\, e.} $\left[ \mathcal{G}_2 - \mathcal{B}_{2_{\alpha i}} \right] \geqslant 0$  is a sufficient condition to ensure the sensitivity of the proposed V2X DS. 
   
    We note here that Sylvester's criterion for $\left[ \mathcal{G}_2 - \mathcal{B}_{2_{\alpha i}} \right] \geqslant 0$  requires $\mathfrak{L}_\alpha$ to be chosen such that  ${\gamma_1}_i = \rho_{3_i}$. Now, if we had not introduced $\rho_{3_i}, \rho_{4_i}$ in \eqref{sk}, it would require  ${\gamma_1}_i= 1$ which significantly constrains the design freedom.  Therefore, to attain more design flexibility, we introduced parameters $\rho_{3_i}$, $\rho_{4_i}$ in place of $\rho_{2_i}$ using \mbox{Table \ref{tab:th1_var}}.

    \textbf{Optimization:} So far, we proved that satisfying the criteria \eqref{iss_con}-\eqref{sen_con} ensures  DSS, robustness, and sensitivity of the proposed detector. Next, to achieve  an optimal trade-off between robustness and sensitivity, we utilize the constrained multi-objective optimization problem established in \eqref{opt} that optimally minimizes  $\rho_{1_i}$ while maximizing $\rho_{2_i}$  over $\mathfrak{L}_\alpha$. The choice of weight vector $\mu_i$  depends on system-specific robustness and sensitivity requirements. The optimizer will be more biased towards sensitivity than robustness for $\mu_i < 0.5$ and vice versa.   The constraints of the optimization problem, i.\,e., the LMIs \eqref{iss_con}, \eqref{rob_con}, and \eqref{sen_con} are respectively the sufficient conditions for DSS, robustness, and sensitivity of the detection scheme. Therefore, the solution of the optimization problem established in \eqref{opt} gives the gain matrices that ensure the benchmark performance of the proposed V2X DS. \hfill  $\blacksquare$


\bibstyle{arxiv}
\bibliography{reference.bib}

\begin{thebibliography}{10}

\bibitem{litman2017autonomous}
Todd Litman.
\newblock {\em Autonomous vehicle implementation predictions}.
\newblock Victoria Transport Policy Institute Victoria, BC, Canada, 2017.

\bibitem{dey2015review}
Kakan~C Dey, Li~Yan, Xujie Wang, Yue Wang, Haiying Shen, Mashrur Chowdhury, Lei Yu, Chenxi Qiu, and Vivekgautham Soundararaj.
\newblock A review of communication, driver characteristics, and controls aspects of cooperative adaptive cruise control ({CACC}).
\newblock {\em IEEE Transactions on Intelligent Transportation Systems}, 17(2):491--509, 2015.

\bibitem{usdata}
Connected vehicles and cybersecurity, 2023.

\bibitem{wyk2020}
Franco van Wyk, Yiyang Wang, Anahita Khojandi, and Neda Masoud.
\newblock Real-time sensor anomaly detection and identification in automated vehicles.
\newblock {\em IEEE Transactions on Intelligent Transportation Systems}, 21(3):1264--1276, 2020.

\bibitem{vanet}
Ian Ku, You Lu, Mario Gerla, Rafael~L. Gomes, Francesco Ongaro, and Eduardo Cerqueira.
\newblock Towards software-defined {VANET}: Architecture and services.
\newblock In {\em 2014 13th Annual Mediterranean Ad Hoc Networking Workshop (MED-HOC-NET)}, pages 103--110, 2014.

\bibitem{buinevich2019forecasting}
Mikhail Buinevich and Andrei Vladyko.
\newblock Forecasting issues of wireless communication networks’ cyber resilience for an intelligent transportation system: An overview of cyber attacks.
\newblock {\em Information}, 10(1):27, 2019.

\bibitem{survey}
Zhiyang Ju, Hui Zhang, Xiang Li, Xiaoguang Chen, Jinpeng Han, and Manzhi Yang.
\newblock A survey on attack detection and resilience for connected and automated vehicles: From vehicle dynamics and control perspective.
\newblock {\em IEEE Transactions on Intelligent Vehicles}, 7(4):815--837, 2022.

\bibitem{sanchita_CCTA2023security}
Sanchita Ghosh and Tanushree Roy.
\newblock Security of cyber-physical systems under compromised switching.
\newblock In {\em 2023 IEEE Conference on Control Technology and Applications (CCTA)}, pages 1034--1039. IEEE, 2023.

\bibitem{swaroop1998intelligent}
DVAHG Swaroop and R~Huandra.
\newblock Intelligent cruise control system design based on a traffic flow specification.
\newblock {\em Vehicle system dynamics}, 30(5):319--344, 1998.

\bibitem{xu2003simulation}
Qing Xu and Raja Sengupta.
\newblock Simulation, analysis, and comparison of {ACC} and {CACC} in highway merging control.
\newblock In {\em IEEE IV2003 Intelligent Vehicles Symposium. Proceedings (Cat. No. 03TH8683)}, pages 237--242. IEEE, 2003.

\bibitem{balador2022survey}
Ali Balador, Alessandro Bazzi, Unai Hernandez-Jayo, Idoia de~la Iglesia, and Hossein Ahmadvand.
\newblock A survey on vehicular communication for cooperative truck platooning application.
\newblock {\em Vehicular Communications}, page 100460, 2022.

\bibitem{song2018multi}
Ke~Song, Feiqiang Li, Xiao Hu, Lin He, Wenxu Niu, Sihao Lu, and Tong Zhang.
\newblock Multi-mode energy management strategy for fuel cell electric vehicles based on driving pattern identification using learning vector quantization neural network algorithm.
\newblock {\em Journal of Power Sources}, 389:230--239, 2018.

\bibitem{lin2004driving}
Chan-Chiao Lin, Soonil Jeon, Huei Peng, and Jang Moo~Lee.
\newblock Driving pattern recognition for control of hybrid electric trucks.
\newblock {\em Vehicle System Dynamics}, 42(1-2):41--58, 2004.

\bibitem{driving}
R.~Langari and Jong-Seob Won.
\newblock Intelligent energy management agent for a parallel hybrid vehicle-part {I}: system architecture and design of the driving situation identification process.
\newblock {\em IEEE Transactions on Vehicular Technology}, 54(3):925--934, 2005.

\bibitem{jeon2002multi}
Soon-il Jeon, Sung-tae Jo, Yeong-il Park, and Jang-moo Lee.
\newblock Multi-mode driving control of a parallel hybrid electric vehicle using driving pattern recognition.
\newblock {\em J. Dyn. Sys., Meas., Control}, 124(1):141--149, 2002.

\bibitem{tweet}
Eleonora D'Andrea, Pietro Ducange, Beatrice Lazzerini, and Francesco Marcelloni.
\newblock Real-time detection of traffic from twitter stream analysis.
\newblock {\em IEEE Transactions on Intelligent Transportation Systems}, 16(4):2269--2283, 2015.

\bibitem{chen2010review}
Bo~Chen and Harry~H Cheng.
\newblock A review of the applications of agent technology in traffic and transportation systems.
\newblock {\em IEEE Transactions on intelligent transportation systems}, 11(2):485--497, 2010.

\bibitem{yan2007providing}
Gongjun~Yan Yan, Gyanesh Choudhary, Michele~C Weigle, and Stephan Olariu.
\newblock Providing {VANET} security through active position detection.
\newblock In {\em Proceedings of the fourth ACM international workshop on Vehicular ad hoc networks}, pages 73--74, 2007.

\bibitem{hasrouny2017vanet}
Hamssa Hasrouny, Abed~Ellatif Samhat, Carole Bassil, and Anis Laouiti.
\newblock {VANet} security challenges and solutions: A survey.
\newblock {\em Vehicular Communications}, 7:7--20, 2017.

\bibitem{sdn}
Mohammad~Ali Salahuddin, Ala Al-Fuqaha, and Mohsen Guizani.
\newblock Software-defined networking for {RSU} clouds in support of the internet of vehicles.
\newblock {\em IEEE Internet of Things Journal}, 2(2):133--144, 2015.

\bibitem{dsrc}
John~B. Kenney.
\newblock Dedicated short-range communications ({DSRC}) standards in the united states.
\newblock {\em Proceedings of the IEEE}, 99(7):1162--1182, 2011.

\bibitem{bai2010toward}
Fan Bai, Daniel~D Stancil, and Hariharan Krishnan.
\newblock Toward understanding characteristics of dedicated short-range communications ({DSRC}) from a perspective of vehicular network engineers.
\newblock In {\em Proceedings of the sixteenth annual international conference on Mobile computing and networking}, pages 329--340, 2010.

\bibitem{wang2020modeling}
Pengcheng Wang, Xinkai Wu, and Xiaozheng He.
\newblock Modeling and analyzing cyberattack effects on connected automated vehicular platoons.
\newblock {\em Transportation research part C: emerging technologies}, 115:102625, 2020.

\bibitem{gunter2020commercially}
George Gunter, Derek Gloudemans, Raphael~E Stern, Sean McQuade, Rahul Bhadani, Matt Bunting, Maria~Laura Delle~Monache, Roman Lysecky, Benjamin Seibold, Jonathan Sprinkle, et~al.
\newblock Are commercially implemented adaptive cruise control systems string stable?
\newblock {\em IEEE Transactions on Intelligent Transportation Systems}, 22(11):6992--7003, 2020.

\bibitem{engoulou2014vanet}
Richard~Gilles Engoulou, Martine Bella{\"\i}che, Samuel Pierre, and Alejandro Quintero.
\newblock {VANET} security surveys.
\newblock {\em Computer Communications}, 44:1--13, 2014.

\bibitem{al2016security}
Ahmed~Shoeb Al~Hasan, Md~Shohrab Hossain, and Mohammed Atiquzzaman.
\newblock Security threats in vehicular ad hoc networks.
\newblock In {\em 2016 international conference on advances in computing, communications and informatics (ICACCI)}, pages 404--411. IEEE, 2016.

\bibitem{roy2021cyber}
Tanushree Roy, Sara Sattarzadeh, and Satadru Dey.
\newblock Cyber-attack detection in socio-technical transportation systems exploiting redundancies between physical and social data.
\newblock {\em arXiv preprint arXiv:2103.11422}, 2021.

\bibitem{reilly2016creating}
Jack Reilly, S{\'e}bastien Martin, Mathias Payer, and Alexandre~M Bayen.
\newblock Creating complex congestion patterns via multi-objective optimal freeway traffic control with application to cyber-security.
\newblock {\em Transportation Research Part B: Methodological}, 91:366--382, 2016.

\bibitem{ghosh2023security}
Sanchita Ghosh and Tanushree Roy.
\newblock Security of cyber-physical systems under compromised switching.
\newblock In {\em 2023 IEEE Conference on Control Technology and Applications (CCTA)}, pages 1034--1039. IEEE, 2023.

\bibitem{islam2018cybersecurity}
Mhafuzul Islam, Mashrur Chowdhury, Hongda Li, and Hongxin Hu.
\newblock Cybersecurity attacks in vehicle-to-infrastructure applications and their prevention.
\newblock {\em Transportation research record}, 2672(19):66--78, 2018.

\bibitem{v2vattack}
Srivalli Boddupalli, Ashwini Hegde, and Sandip Ray.
\newblock Replace: Real-time security assurance in vehicular platoons against v2v attacks.
\newblock In {\em 2021 IEEE International Intelligent Transportation Systems Conference (ITSC)}, pages 1179--1185, 2021.

\bibitem{boddupalli2022resilient}
Srivalli Boddupalli, Akash~Someshwar Rao, and Sandip Ray.
\newblock Resilient cooperative adaptive cruise control for autonomous vehicles using machine learning.
\newblock {\em IEEE Transactions on Intelligent Transportation Systems}, 23(9):15655--15672, 2022.

\bibitem{mashrur}
Gurcan Comert, Mizanur Rahman, Mhafuzul Islam, and Mashrur Chowdhury.
\newblock Change point models for real-time cyber attack detection in connected vehicle environment.
\newblock {\em IEEE Transactions on Intelligent Transportation Systems}, 23(8):12328--12342, 2022.

\bibitem{han2023secure}
Jinpeng Han, Zhiyang Ju, Xiaoguang Chen, Manzhi Yang, Hui Zhang, and Rouxing Huai.
\newblock Secure operations of connected and autonomous vehicles.
\newblock {\em IEEE Transactions on Intelligent Vehicles}, 2023.

\bibitem{kalinin2017network}
Maxim Kalinin, Vasiliy Krundyshev, Peter Zegzhda, and Viacheslav Belenko.
\newblock Network security architectures for {VANET}.
\newblock In {\em Proceedings of the 10th International Conference on Security of Information and Networks}, pages 73--79, 2017.

\bibitem{dataencrypt}
Feng Jiang, Buren Qi, Tianhao Wu, Konglin Zhu, and Lin Zhang.
\newblock {CPSS}: {CP-ABE} based platoon secure sensing scheme against {Cyber-Attacks}.
\newblock In {\em 2019 IEEE Intelligent Transportation Systems Conference (ITSC)}, pages 3218--3223, 2019.

\bibitem{xiao2006detection}
Bin Xiao, Bo~Yu, and Chuanshan Gao.
\newblock Detection and localization of sybil nodes in {VANETs}.
\newblock In {\em Proceedings of the 2006 workshop on Dependability issues in wireless ad hoc networks and sensor networks}, pages 1--8, 2006.

\bibitem{trustdata}
Seyhan Ucar, Sinem~Coleri Ergen, and Oznur Ozkasap.
\newblock Data-driven abnormal behavior detection for autonomous platoon.
\newblock In {\em 2017 IEEE Vehicular Networking Conference (VNC)}, pages 69--72, 2017.

\bibitem{mundhenk2017security}
Philipp Mundhenk, Andrew Paverd, Artur Mrowca, Sebastian Steinhorst, Martin Lukasiewycz, Suhaib~A Fahmy, and Samarjit Chakraborty.
\newblock Security in automotive networks: Lightweight authentication and authorization.
\newblock {\em ACM Transactions on Design Automation of Electronic Systems (TODAES)}, 22(2):1--27, 2017.

\bibitem{yan2015software}
Qiao Yan, F~Richard Yu, Qingxiang Gong, and Jianqiang Li.
\newblock Software-defined networking ({SDN}) and distributed denial of service ({DDoS}) attacks in cloud computing environments: A survey, some research issues, and challenges.
\newblock {\em IEEE communications surveys \& tutorials}, 18(1):602--622, 2015.

\bibitem{javed2020cnn}
Abdul~Rehman Javed, Muhammad Usman, Saif~Ur Rehman, Mohib~Ullah Khan, and Mohammad~Sayad Haghighi.
\newblock Anomaly detection in automated vehicles using multistage attention-based convolutional neural network.
\newblock {\em IEEE Transactions on Intelligent Transportation Systems}, 22(7):4291--4300, 2020.

\bibitem{rmltrust}
Jingjing Guo, Xinghua Li, Zhiquan Liu, Jianfeng Ma, Chao Yang, Junwei Zhang, and Dapeng Wu.
\newblock {TROVE}: A context-awareness trust model for {VANET}s using reinforcement learning.
\newblock {\em IEEE Internet of Things Journal}, 7(7):6647--6662, 2020.

\bibitem{bangui2022hybrid}
Hind Bangui, Mouzhi Ge, and Barbora Buhnova.
\newblock A hybrid machine learning model for intrusion detection in {VANET}.
\newblock {\em Computing}, 104(3):503--531, 2022.

\bibitem{close}
Yuanzhe Wang, Qipeng Liu, Ehsan Mihankhah, Chen Lv, and Danwei Wang.
\newblock Detection and isolation of sensor attacks for autonomous vehicles: Framework, algorithms, and validation.
\newblock {\em IEEE Transactions on Intelligent Transportation Systems}, 23(7):8247--8259, 2022.

\bibitem{stateestimation}
Jiaping Xiao and Mir Feroskhan.
\newblock Cyber attack detection and isolation for a quadrotor {UAV} with modified sliding innovation sequences.
\newblock {\em IEEE Transactions on Vehicular Technology}, 71(7):7202--7214, 2022.

\bibitem{kremer2020kf}
Philipp Kremer, Ipsita Koley, Soumyajit Dey, and Sangyoung Park.
\newblock State estimation for attack detection in vehicle platoon using {VANET} and controller model.
\newblock In {\em 2020 IEEE 23rd International Conference on Intelligent Transportation Systems (ITSC)}, pages 1--8. IEEE, 2020.

\bibitem{kf}
Yiyang Wang, Neda Masoud, and Anahita Khojandi.
\newblock Anomaly detection in connected and automated vehicles using an augmented state formulation.
\newblock In {\em 2020 Forum on Integrated and Sustainable Transportation Systems (FISTS)}, pages 156--161, 2020.

\bibitem{he2021distributed}
Xingkang He, Ehsan Hashemi, and Karl~H Johansson.
\newblock Distributed control under compromised measurements: Resilient estimation, attack detection, and vehicle platooning.
\newblock {\em Automatica}, 134:109953, 2021.

\bibitem{estimation}
Eman Mousavinejad, Fuwen Yang, Qing-Long Han, Xiaohua Ge, and Ljubo Vlacic.
\newblock Distributed cyber attacks detection and recovery mechanism for vehicle platooning.
\newblock {\em IEEE Transactions on Intelligent Transportation Systems}, 21(9):3821--3834, 2020.

\bibitem{roy2021socio}
Tanushree Roy, Amara Tariq, and Satadru Dey.
\newblock A socio-technical approach for resilient connected transportation systems in smart cities.
\newblock {\em IEEE Transactions on Intelligent Transportation Systems}, 23(6):5019--5028, 2021.

\bibitem{yang2021secure}
Tianci Yang and Chen Lv.
\newblock Secure estimation and attack isolation for connected and automated driving in the presence of malicious vehicles.
\newblock {\em IEEE Transactions on Vehicular Technology}, 70(9):8519--8528, 2021.

\bibitem{biron2018real}
Zoleikha~Abdollahi Biron, Satadru Dey, and Pierluigi Pisu.
\newblock Real-time detection and estimation of denial of service attack in connected vehicle systems.
\newblock {\em IEEE Transactions on Intelligent Transportation Systems}, 19(12):3893--3902, 2018.

\bibitem{zhai2018switched}
Chunjie Zhai, Yonggui Liu, and Fei Luo.
\newblock A switched control strategy of heterogeneous vehicle platoon for multiple objectives with state constraints.
\newblock {\em IEEE Transactions on Intelligent Transportation Systems}, 20(5):1883--1896, 2018.

\bibitem{yao2020managing}
Shengyue Yao, Rahi~Avinash Shet, and Bernhard Friedrich.
\newblock Managing connected automated vehicles in mixed traffic considering communication reliability: a platooning strategy.
\newblock {\em Transportation Research Procedia}, 47:43--50, 2020.

\bibitem{turri2016cooperative}
Valerio Turri, Bart Besselink, and Karl~H Johansson.
\newblock Cooperative look-ahead control for fuel-efficient and safe heavy-duty vehicle platooning.
\newblock {\em IEEE Transactions on Control Systems Technology}, 25(1):12--28, 2016.

\bibitem{roy2020secure}
Tanushree Roy and Satadru Dey.
\newblock Secure traffic networks in smart cities: Analysis and design of cyber-attack detection algorithms.
\newblock In {\em 2020 American Control Conference (ACC)}, pages 4102--4107. IEEE, 2020.

\bibitem{hespanha1999stability}
Joao~P Hespanha and A~Stephen Morse.
\newblock Stability of switched systems with average dwell-time.
\newblock In {\em Proceedings of the 38th IEEE conference on decision and control (Cat. No. 99CH36304)}, volume~3, pages 2655--2660. IEEE, 1999.

\bibitem{duo2022survey}
Wenli Duo, MengChu Zhou, and Abdullah Abusorrah.
\newblock A survey of cyber attacks on cyber physical systems: Recent advances and challenges.
\newblock {\em IEEE/CAA Journal of Automatica Sinica}, 9(5):784--800, 2022.

\bibitem{chowdhury2020security}
Mashrur Chowdhury, Mhafuzul Islam, and Zadid Khan.
\newblock Security of connected and automated vehicles.
\newblock {\em arXiv preprint arXiv:2012.13464}, 2020.

\bibitem{dingbook}
Steven~X Ding.
\newblock Model-based fault diagnosis techniques: Design schemes, algorithms, and tools.

\bibitem{dehghan2016sensorless}
Ehsan Dehghan-Azad, Shady Gadoue, David Atkinson, Howard Slater, Peter Barrass, and Frede Blaabjerg.
\newblock Sensorless control of im for limp-home mode ev applications.
\newblock {\em IEEE Transactions on Power Electronics}, 32(9):7140--7150, 2016.

\bibitem{ploeg2011design}
Jeroen Ploeg, Bart~TM Scheepers, Ellen Van~Nunen, Nathan Van~de Wouw, and Henk Nijmeijer.
\newblock Design and experimental evaluation of cooperative adaptive cruise control.
\newblock In {\em 2011 14th International IEEE Conference on Intelligent Transportation Systems (ITSC)}, pages 260--265. IEEE, 2011.

\end{thebibliography}

\end{document}